\documentclass{aptpub}

\usepackage{bbm}
\usepackage{graphicx}
\usepackage{subcaption}
\usepackage{tikz}
\usepackage{algorithm}
\usepackage{savesym}
\savesymbol{OR}
\usepackage{algorithmic}
\usepackage{soul}

\setulcolor{red}

\renewcommand{\E}{\mathbb{E}}
\renewcommand{\N}{\mathbb{N}}
\renewcommand{\P}{\mathbb{P}}
\newcommand{\R}{\mathbb{R}}

\newcommand{\Var}{\operatorname{Var}}
\newcommand{\1}{\mathbbm{1}}

\newcommand{\half}{\mbox{$\frac 1 2 $}}

\authornames{JORIS BIERKENS, ANDREW DUNCAN} 
\shorttitle{Limit Theorems for the Zig-Zag Process} 

\begin{document}

\title{Limit Theorems for the Zig-Zag Process} 

\authorone[University of Warwick]{Joris Bierkens} 
\authortwo[Imperial College]{Andrew Duncan}

\addressone{Delft Institute of Applied Mathematics, Mekelweg 4, 2628 CD, Delft, Netherlands} 
\addresstwo{Department of Mathematics, University of Sussex, Brighton BN1 9QH, United Kingdom}



\begin{abstract}
Markov chain Monte Carlo methods provide an essential tool in statistics for sampling from complex probability distributions.  While the standard approach to MCMC involves constructing discrete-time reversible Markov chains whose transition kernel is obtained via the Metropolis-Hastings algorithm, there has been recent interest in alternative schemes based on piecewise deterministic Markov processes (PDMPs).  One such approach is based on the Zig-Zag process, introduced in \cite{bierkensroberts2015}, which proved to provide a highly scalable sampling scheme for sampling in the big data regime \cite{BierkensFearnheadRoberts2016}. In this paper we study the performance of the Zig-Zag sampler, focusing on the one-dimensional case.  In particular, we identify conditions under which a Central limit theorem (CLT) holds and characterize the asymptotic variance.  Moreover, we study the influence of the switching rate on the diffusivity of the Zig-Zag process by identifying a diffusion limit as the switching rate tends to infinity. Based on our results we compare the performance of the Zig-Zag sampler to existing Monte Carlo methods, both analytically and through simulations.
\end{abstract}


\section{Introduction}

 Markov Chain Monte Carlo methods remain an essential computational tool in statistics and among other things have made it possible for Bayesian inference techniques to be applied to increasingly complex models.   Due to its simplicity and wide applicability, the Metropolis-Hastings (MH) algorithm \cite{metropolis1953equation,hastings1970monte} and its numerous variants remain the most widely used MCMC method for sampling from a general target probability distribution, despite having been introduced over 60 years ago. Given a target distribution $\pi$, the Metropolis-Hastings scheme defines a discrete time Markov chain which will be both ergodic and reversible with respect to $\pi$.  The fact that the Markov chain is reversible is a serious limitation.  Indeed, it is now well known that non-reversible chains can significantly outperform reversible chains, in terms of rate of convergence to equilibrium \cite{hwang1993accelerating,lelievre2013optimal}, asymptotic variance \cite{chen2013accelerating,sun2010improving,duncan2016variance} as well as large deviation functionals \cite{rey2015variance,rey2015irreversible,rey2016improving}.    One particular approach to improving performance is to introduce a velocity/momentum variable and construct Markovian dynamics which are able to mixing more rapidly in the augmented state space.  Such methods include Hybrid Monte Carlo (HMC) methods, inspired by Hamiltonian dynamics, and numerous generalisations.  While the standard construction of HMC \cite{duane1987hybrid,neal2011mcmc} is reversible,  it is straightforward to alter the scheme such that the resulting process is non-reversible \cite{ottobre2016function}.

 In \cite{bierkensroberts2015}, the Zig-Zag process was introduced, a continuous time piecewise deterministic process (PDMP) which provides a practical sampling scheme applicable for a wide class of probability distributions.  Given a target density $\pi$, known up to a multiplicative constant,  the one dimensional Zig-Zag process is a continuous time Markov process $(X(t), \Theta(t))_{t\geq 0}$ on $E = \mathbb{R}\times \lbrace -1, +1 \rbrace$, such that $X(t)$ moves with constant velocity $\Theta(t)$. The velocity process $\Theta(t)$ switches its values between $-1$ and $+1$ at random times obtained from a inhomogeneous Poisson process with switching rate $\lambda(X(t),\Theta(t))$. If the switching rate is chosen to agree with the target distribution $\pi$ in a certain way, this guarantees that the Zig-Zag process has stationary distribution $\mu$ on $\mathbb R \times \{-1,+1\}$, whose marginal distribution on $\R$ is proportional to $\pi$. As a consequence, the law of large numbers,
 \begin{equation} \label{eq:ergodic_average_intro} \E_{\pi} [f] = \int_{\R} f(x) \pi(x) \, d x = \lim_{T \rightarrow \infty} \frac 1 T \int_0^T f(X(s)) \, d s, \end{equation}
 is satisfied, so that the Zig-Zag process can be used to approximate expectations with respect to $\pi$. Two one-dimensional examples of the Zig-Zag process are displayed in Figure~\ref{fig:examples-introduction}.
 
\begin{figure}
\begin{subfigure}[b]{0.45 \textwidth}
 \includegraphics[width=\textwidth]{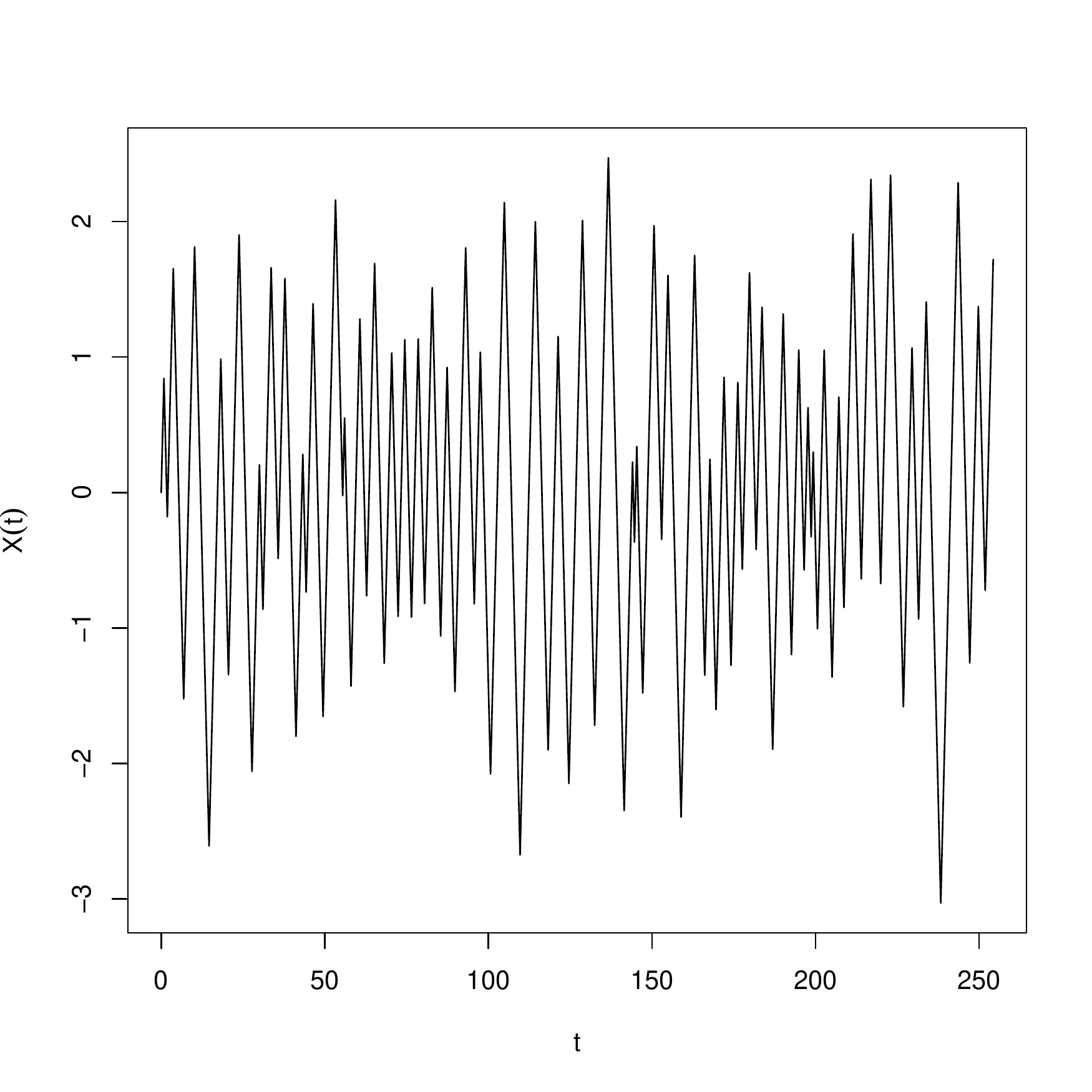}
 \caption{1D Gaussian}
\end{subfigure}
\begin{subfigure}[b]{0.45 \textwidth}
 \includegraphics[width=\textwidth]{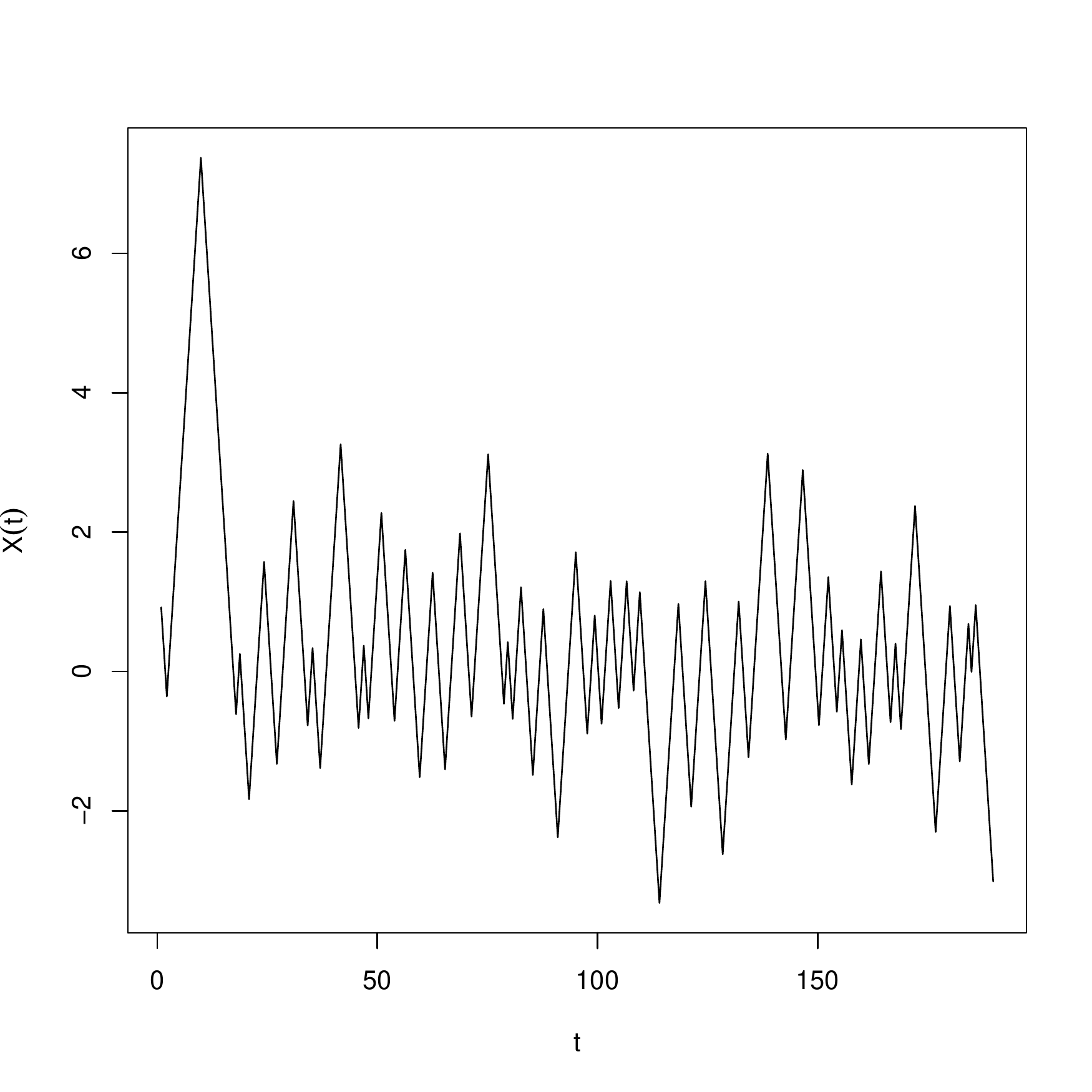}
 \caption{1D Cauchy}
\end{subfigure}

\caption{Example trajectories of the Zig-Zag process with the specified invariant distributions.} 
\label{fig:examples-introduction}
\end{figure}

While the construction and finite-time behaviour of PDMPs is well understood \cite{Davis1984}, their use within the context of sampling has only recently been  considered and is mostly unexplored.  The first such occurrence of a MCMC scheme based on PDMP appeared in the computational physics literature \cite{peters2012rejection} and in one dimension coincides with the Zig-Zag sampler.  This scheme was extended and analysed carefully in \cite{bouchard2015bouncy}, where it was rechristened the Bouncy Particle Sampler.  In one dimension, the quantitative long-time behaviour of related PDMP schemes has been analysed in detail, see for example \cite{azais2014piecewise,fontbona2012quantitative,fontbona2016long,monmarche2015mathcal,monmarche2014piecewise}.
More recently in \cite{BierkensFearnheadRoberts2016}, the application of the Zig-Zag sampler to big data settings was investigated. It was found that the Zig Zag sampler lends itself very well to such problems since sub-sampling can be introduced without affecting the stationary distribution, as opposed to standard sub-sampling techniques, such as SGLD \cite{welling2011bayesian} which are inherently biased.   By introducing appropriate control variates a ``super-efficient'' sampling scheme for big data problems was produced,  in the sense that it is able to generate independent samples from the target distribution at a higher efficiency than directly generating IID samples using the entire data set for each sample.
 
In this paper we seek to better understand the qualitative performance of the Zig Zag sampler.  Focusing on the one-dimensional case, we study the important practical question of whether a central limit theorem (CLT) holds for the Zig-Zag process, i.e. whether for a given observable $f$,
 \begin{equation}
 \label{eq:CLT}
   \sqrt{t}\left(\frac{1}{t}\int_0^t f(X(s))\,ds - \mathbb{E}_{\pi}[f] \right)\Rightarrow \mathcal{N}(0,\sigma_f^2),\quad \mbox{ as } t \rightarrow \infty,
 \end{equation}
where $\sigma^2_f$ is the asymptotic variance and where $\Rightarrow$ denotes convergence in distribution. Heuristically, once a CLT is known to hold, we know that the ergodic average in~\eqref{eq:ergodic_average_intro} converges at rate $\sigma_f/\sqrt{t}$, which is the best convergence to be expected in a Monte Carlo simulation. It is also clear that a smaller value of $\sigma_f > 0$ implies a faster convergence of the ergodic averages. Without a CLT, convergence may be arbitrarily slow.  Starting from the case of a unimodal target distribution and extending to more general cases, we obtain sufficient conditions for~\eqref{eq:CLT} to hold.   Moreover, we identify conditions under with the CLT can be strengthened to an invariance principle or functional central limit theorem (FCLT) \cite{Komorowski2012}.  For the one-dimensional Zig-Zag process we obtain explicit expressions for the asymptotic variance, which we illustrate for various examples.

Given a target distribution $\pi$, there is some freedom in choosing the switching rate $\lambda$ in such a way that $\pi$ is invariant for the Zig-Zag process. This freedom is crucial for the ability of the sub-sampling Zig-Zag scheme of \cite{BierkensFearnheadRoberts2016} to sample without bias.  In Section \ref{sec:diffusive} we study the influence of the particular choice of switching rate on the behaviour of the process.  We show that as the switching rate is increased the Zig-Zag sampler will exhibit random walk behaviour. In particular, over an appropriate timescale the Zig-Zag sampler will behave asymptotically, as the excess switching rate tends to infinity, as an overdamped Langevin diffusion which is ergodic with respect to $\pi$.   

As the Zig-Zag sampler is based upon a continuous time process, it is not immediately clear how its performance can be compared to existing discrete time sampling schemes.  With this aim in mind, we derive approximations for the average switching rate of the process per unit time, and apply this to construct an \emph{effective sample size (ESS)} for the Zig-Zag sampler which quantifies the number of independent samples generated  in terms of the number of evaluations of the gradient of the log density. A suitable definition of effective sample size depends in an essential way on the asymptotic variance of the corresponding CLT, which further illustrates the importance of establishing a CLT from an applied viewpoint. Comparing to IID samples in some cases we observe a remarkable feature: the effective sample size of the Zig-Zag sampler will be larger than that of IID samples, behaviour which is strongly tied to the nonreversibility of the scheme.

We structure the paper as follows.  In Section~\ref{sec:zig-zag} we review the construction of the Zig-Zag sampler in the one dimensional case and explore its basic properties.  Section~\ref{sec:clt} describes conditions for a CLT to hold for the one dimensional Zig-Zag sampler and characterises the asymptotic variance.  These results are demonstrated numerically for some standard probability distributions.  In Section~\ref{sec:diffusive}  the diffusive regime is investigated where the switching rate $\lambda$ goes to infinity.  Finally, in Section~\ref{sec:ess} an appropriate measure of effective sample size is introduced for the Zig-Zag sampler, and is used to compare the performance of the Zig-Zag sampler with other sampling techniques for some standard probability distributions. The proofs of most of results may be found in Appendix A. In Appendix B we discuss the simulation of the Zig-Zag process, which provides the necessary background for Section~\ref{sec:ess}.



\subsection{Notation}


For $E$ a topological space, the space of continuous functions $f : E \rightarrow \mathbb R$ is denoted by $C(E)$, and $\mathcal M(E)$ denotes the set of Borel measurable functions on $E$. The Borel sets in $E$ are denoted by $\mathcal B(E)$. On a measurable space $E$, the measure $\delta_x$, for $x \in E$, is defined as the probability measure assigning mass $1$ to $x$.  Lebesgue measure on $\R^d$ is denoted by $\mathrm{Leb}$. 
The Skorohod space of cadlag paths from an interval $I \subset \R$ into $E$ is denoted by $D(I;E)$; see \cite{EthierKurtz2005} for details. The Skorohod space of cadlag paths from $I$ into $\R$ is also denoted by $D(I)$. We use the symbol $\Rightarrow$ to indicate weak convergence of probability distributions, where the relevant topology (either the natural topology on $\R$ or the Skorohod topology on the space of cadlag paths) can be deduced from the context. We write $\mathcal L(X)$ for the law of a random variable $X$. The pushforward $\mu_{\star} f$ of a measure $\mu$ on $E$ by a measurable function $f : E \rightarrow F$, with $E$ and $F$ measurable spaces, is defined as $\mu_{\star} f(A) := \mu(f^{-1}(A))$ for measurable sets $A$ in $F$.
We write $\Phi$ for the cumulative distribution function of the standard normal distribution. We will use the notation $\pi$ for a probability density function $\pi : \R \rightarrow [0,\infty)$, as well as for the associated probability measure, so e.g. $\pi(f) = \int_{\R} f(x) \pi(x) \ d x$. For $a \in \R$ we will write $(a)^+$ and $(a)^-$ for the positive and negative parts of $a$ respectively, i.e.  $(a)^+ = \max(0,a)$ and $(a)^- = \max(0, -a)$. 


\section{The Zig-Zag process}
\label{sec:zig-zag}
In this section we review some earlier established results on the Zig-Zag process.
Let $E = \R \times \{-1,+1\}$ and equip $E$ with the product topology of open sets in $\R$ and the discrete topology on $\{-1,+1\}$.
The following assumption will be sufficient to define the Zig-Zag process, and ensure it has a unique invariant distribution.

\begin{assumption}
 \label{ass:well-posedness-stationary-zigzag}
$\lambda : E \rightarrow \R_+$ is continuous and the function
\begin{equation} \label{eq:lambda-implicit} U(x) := \int_0^x \{ \lambda(\xi,+1) - \lambda(\xi,-1) \} \ d \xi \end{equation}
satisfies
\[ \int_{-\infty}^{\infty} \exp(-U(x)) \ d x < \infty.\]
Furthermore for some $x_0 > 0$, we have $\lambda(x,\theta) > 0$ if $\theta x \geq x_0$.
\end{assumption}

An alternative and convenient way of writing~\eqref{eq:lambda-implicit} is $\lambda(x,\theta) - \lambda(x,-\theta) = \theta U'(x)$ for all $(x,\theta) \in E$.
It is easy to check that~\eqref{eq:lambda-implicit} holds if and only if there exists a continuously differentiable function $U$ and a continuous non-negative function $\gamma$ such that
\begin{equation} \label{eq:lambda-explicit} \lambda(x,\theta) = \max(0, \theta U'(x)) + \gamma(x).\end{equation}
The switching rates $\lambda$ for which $\gamma \equiv 0$ are called \emph{canonical switching rates} and the corresponding Zig-Zag process is called the \emph{canonical Zig-Zag process}.

Let $\nu$ denote a reference measure on $E$ given by $\nu := \mathrm{Leb} \otimes (\delta_{-1} + \delta_{+1})$. We use $\nu$ to define the probability measure $\mu$ by
\[ \frac{d \mu}{d \nu}(x,\theta) = \frac{\exp(-U(x))}{2k}, \quad (x, \theta) \in E,\]
where $k := \int_{\R} \exp(-U(x)) \ d x$.
The marginal distribution of $\mu$ with respect to $x$ has Lebesgue density proportional to $\exp(-U(x))$, denoted by $\pi$, i.e. $\pi(x) = \exp(-U(x))/k$. 

Define an operator $L$ with domain
\[ \mathcal D(L) = \{ f \in C(E) : \mbox{$f(\cdot, \theta)$ is absolutely continuous for $\theta = \pm 1$} \}\]
by
\begin{equation} \label{eq:generator-zigzag} L f(x,\theta) = \theta \partial_x f(x,\theta) + \lambda(x,\theta) (f(x, -\theta) - f(x,\theta)), \quad (x, \theta) \in E,\end{equation} which will service as the generator of the Markov semigroup of the Zig-Zag process, with dynamics as discussed in the introduction.
In the following proposition, the notion of `petite sets' can be found in \cite{MeynTweedie1993-II}.

\begin{proposition} 
\label{prop:well-posedness-stationary-zigzag} Suppose Assumption~\ref{ass:well-posedness-stationary-zigzag} holds. Then $(L, \mathcal{D}(L))$ is the extended generator of a piecewise deterministic Markov-Feller process $(Z(t))_{t \geq 0} := (X(t), \Theta(t))_{t \geq 0}$ in $E$. All compact sets are petite for $(X(t),\Theta(t))$. Finally $\mu$ is the unique invariant probability distribution for $(Z(t))_{t \geq 0}$.
\end{proposition}

The proof of this result is located in Appendix~\ref{app:well-posedness}.

The above setting can be used for Monte Carlo sampling as follows. Starting from a normalizable (but possibly unnormalized), strictly positive and continuously differentiable density $\widetilde \pi(x)$ on $\R$, we can define $U(x) := - \log \widetilde \pi(x)$, and define $\lambda(x,\theta)$ by~\eqref{eq:lambda-explicit} for some non-negative function $\gamma$ of our choice. Assuming that, for some $x_0 > 0$, either $\gamma(x) > 0$ for $|x| \geq x_0$, or that $\theta U'(x) > 0$ for $\theta x \geq x_0$, Assumption~\ref{ass:well-posedness-stationary-zigzag} is satisfied, and the process constructed in Proposition~\ref{prop:well-posedness-stationary-zigzag} has marginal stationary distribution $\pi$ on $\R$, where $\pi$ is the normalization of $\widetilde \pi$.

We call $(Z(t))_{t \geq 0} = (X(t), \Theta(t))_{t \geq 0}$ the \emph{Zig-Zag process} with \emph{switching intensity} $\lambda(x,\theta)$. Although the paths of the Zig-Zag process are continuous in $E$, in view of our goal of obtaining limit theorems for the Zig-Zag process we will consider its sample paths as elements in $D([0,\infty);E)$. For any probability distribution $\eta$ on $E$ let $\P_{\eta}$ denote the probability measure on $D([0,\infty);E)$ for the Zig-Zag process with initial distribution $\eta$. In particular under $\P_{\mu}$ the law of $(Z(t))_{t \geq 0}$ is stationary.

\section{Central Limit Theorems for the Zig-Zag process}
\label{sec:clt}

First, in Section~\ref{sec:unimodal}, we obtain a CLT for the Zig-Zag process in the simple and intuitive case in which the target distribution is unimodal and the excess switching rate $\gamma = 0$.
Then we describe a general approach to the CLT in Section~\ref{sec:abstract-CLT}. We then illustrate the theory with several examples in Section~\ref{sec:clt_examples}.


\subsection{The CLT for the special case of a unimodal invariant distribution}
\label{sec:unimodal}

If the potential $U(x)$ is continuously differentiable and is monotonically non-decreasing (non-increasing) for $x \geq 0$ ($x \leq 0$) then the canonical switching rates associated with $U$ satisfy $\lambda(x,+1) = 0$ for $x \leq 0$, and $\lambda(x,-1) = 0$ for $x \geq 0$. In this situation trajectories of the canonical Zig-Zag process will always pass through the origin $x = 0$ between switches. This regular behaviour makes it possible to obtain a Central Limit Theorem in a very straightforward way: by inspecting the contributions towards the total variance of trajectory segments between crossings of the origin.

\begin{assumption}
\begin{itemize}
 \item[(i)] $U : \R \rightarrow [0,\infty)$ is continuously differentiable and is monotonically non-decreasing (non-increasing) for $x \geq 0$ ($x \leq 0$). Furthermore $k:= \int_{\R} \exp(-U(x)) \ d x <\infty$;
 \item[(ii)] $g : \R \rightarrow \R$ is integrable with respect to $\pi$ and satisfies $\int_{\R} g(x) \pi(x) \ d x= 0$, where $\pi(x) :=\exp(-U(x))/k$;
 \item[(iii)] We have 
 \[ \int_{\R} |U'(t)| \exp(-U(t)) \left(\int_0^t g(s) \ d s \right)^2 \ d t < \infty.\]
 \item[(iv)] $\lambda(x,\theta)$ are the canonical switching rates defined by $\lambda(x,\theta) = (\theta U'(x))^+$.
%
%
\end{itemize}
\label{ass:unimodal}
\end{assumption}

Note that the definition of $\pi$ agrees with the definition of $\pi$ below Assumption~\ref{ass:well-posedness-stationary-zigzag}. Furthermore, the fact that $\exp(-U(x))$ is integrable, combined with the monotonicity assumption, implies that the switching rates $\lambda(x,\theta)$ are positive for $\theta x \geq x_0$, for some fixed $x_0 > 0$, so that Assumption~\ref{ass:unimodal} implies Assumption~\ref{ass:well-posedness-stationary-zigzag}.

\begin{theorem}
\label{thm:unimodal-CLT}
Suppose Assumption~\ref{ass:unimodal} holds. Let $(X(t),\Theta(t))$ denote the Zig-Zag process with switching rates $\lambda(x,\theta)$. Then 
\[ \frac 1 {\sqrt t} \int_0^t g(X(s)) \ d s \Rightarrow \mathcal N(0, \sigma_g^2),\]
where
\begin{align}
 \label{eq:asvar-alternative} \sigma_g^2 &  := \frac{ 2 \int_{\R} |U'(t)| \exp(-U(t)) \left( \int_0^t g(s) \ ds \right)^2 \ d t - 4 \left(\int_0^{\infty}  \exp(-U(t)) g(t)  \ d t \right)^2}{\int_{-\infty}^{\infty} \exp(-U(t)) \ d t }
\end{align}
\end{theorem}

\begin{proof}
Iteratively define random times $(T_i^{\pm})_{i \in \N}$ and $(S_i^{\pm})_{i \in \N}$ as follows:
\begin{align*} T_0^+ & = \inf\{ t \geq 0: X(t) = 0 \ \mbox{and} \ \Theta(t) = +1\}, \\
T_i^- & = \inf\{ t > T_{i-1}^+  : X(t) = 0 \ \mbox{and} \ \Theta(t) = -1\},  \quad & i = 1, 2, 3, \dots, \\
T_i^+ & = \inf\{ t > T_i^- : X(t) = 0 \ \mbox{and} \ \Theta(t) = +1\}, \quad & i = 1, 2, 3, \dots,  \\
S_i^+ & = \inf\{ t > T_{i-1}^+ : \Theta(t) = - 1\}, \quad & i = 1, 2, 3, \dots, \\
S_i^- & = \inf\{ t > T_i^- : \Theta(t) = + 1\}, \quad & i = 1, 2, 3, \dots.
\end{align*}
See Figure~\ref{fig:unimodal-clt-illustration} for a graphical illustration of these times.

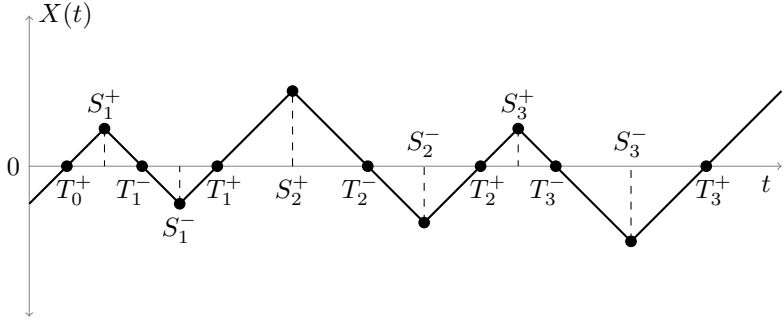
\begin{figure}[!ht]
{\begin{center}
\begin{tikzpicture} 

\draw[help lines, <->] (0,2) -- (0,0) -- (10,0);
\draw[help lines, ->] (0,0) -- (0,-2);

\draw [thick] (0,-0.5) -- (1,0.5);
\draw [thick] (1,0.5) -- (2,-0.5);
\draw [thick] (2,-0.5) -- (3.5, 1);
\draw [thick] (3.5,1) -- (5.25,-0.75);
\draw [thick] (5.25,-0.75) -- (6.5, 0.5);
\draw [thick] (6.5, 0.5) -- (8, -1);
\draw [thick] (8,-1) -- (10,1);

\node [left] at (0,0) {0};

\draw[fill] (0.5,0) circle [radius =2pt];
\node [below] at (0.6, 0) {$T_0^+$};

\draw[fill] (1.5,0) circle [radius =2pt];
\node [below] at (1.4,0) {$T_1^-$};

\draw[fill] (2.5,0) circle [radius =2pt];
\node [below] at (2.6,0) {$T_1^+$};

\draw[fill] (4.5,0) circle [radius =2pt];
\node [below] at (4.4,0) {$T_2^-$};

\draw[fill] (6.0,0) circle [radius =2pt];
\node [below] at (6.1,0) {$T_2^+$};

\draw[fill] (7.0,0) circle [radius =2pt];
\node [below] at (6.9,0) {$T_3^-$};

\draw[fill] (9.0,0) circle [radius =2pt];
\node [below] at (9.1,0) {$T_3^+$};

\draw[fill] (1,0.5) circle [radius = 2pt];
\draw [dashed] (1,0.5) -- (1,0);
\node [above] at (1,0.5) {$S_1^+$};

\draw[fill] (2,-0.5) circle [radius = 2pt];
\draw [dashed] (2,-0.5) -- (2,0);
\node [below] at (2,-0.5) {$S_1^-$};

\draw[fill] (3.5,1) circle [radius = 2pt];
\draw [dashed] (3.5,0) -- (3.5,1);
\node [below] at (3.5,0) {$S_2^+$};

\draw[fill] (5.25,-0.75) circle [radius = 2pt];
\draw [dashed] (5.25,-0.75) -- (5.25,0);
\node [above] at (5.25,0) {$S_2^-$};

\draw[fill] (6.5,0.5) circle [radius = 2pt];
\draw [dashed] (6.5,0.5) -- (6.5, 0);
\node [above] at (6.5,0.5) {$S_3^+$};

\draw[fill] (8,-1) circle [radius = 2pt];
\draw [dashed] (8,-1) -- (8,0);
\node [above] at (8,0) {$S_3^-$};

\node [below] at (9.8,0) {$t$};
\node [right] at (0, 2) {$X(t)$};

\end{tikzpicture}
\caption{Graphical illustration of the random times $S_i^{\pm}, T_i^{\pm}$ introduced in the proof of Theorem~\ref{thm:unimodal-CLT}.}
\label{fig:unimodal-clt-illustration}  
 \end{center}
 }
\end{figure}

Now for $i=1,2,\ldots$, define the random variables
\begin{align*} Y_i^+ & := \int_{T_{i-1}^+}^{T_{i}^-} g(s) \ d s = 2 \int_{T_{i-1}^+}^{S_{i}^+} g(s) \ d s, \\
 Y_i^- & := 2 \int_{T_i^-}^{T_i^+} g(s) \ d s = 2 \int_{T_{i}^-}^{S_i^-} g(s) \ d s, \quad \mbox{and} \\
 Y_i & := Y_i^+ + Y_i^-.
\end{align*}
Let $N(t) := \sup\{ i : T_i^+ \leq t\}$.
Then
\[ \frac 1 {\sqrt{t}} \int_0^t g(X(s)) \ d s = \frac 1 {\sqrt{t}} \left( \int_0^{T_0^+} g(X(s)) \ d s + \sum_{i=1}^{N(t)} Y_i + \int_{T_{N(t)}^+}^t g(s) \ d s \right).\]
Note that $(Y_i)$ are i.i.d., with distribution identical to that of the random variable $Y := Y^+ + Y^-$, where $Y^+$ and $Y^-$ are independent random variables defined by
\[ Y^+ := 2 \int_0^{\tau^+} g(s) \ d s, \quad Y^- := 2 \int_0^{\tau^-} g(-s) \ d s,\]
where $\P(\tau^{\pm} \geq t) = \exp \left( -\int_0^t \lambda( \pm s,\pm 1) \ d s \right)$.
We compute
\begin{align*}
 \E[Y^+] & = \int_0^{\infty} \lambda(t,+1) \exp \left( -\int_0^t \lambda(s,+1) \, d s \right) \left(2  \int_0^t g(s) \, d s \right) \, d t  \\
 & = 2 \int_0^{\infty} U'(t) \exp(-U(t)) \left( \int_0^t g(s) \ d s \right) \ d t,
\end{align*}
and, using Assumption~\ref{ass:unimodal} (ii),
\begin{align*}
 \E[Y^-] & = \int_0^{\infty} \lambda(-t,-1) \exp \left( -\int_0^t \lambda(-s,-1) \, d s \right) \left( 2 \int_0^t g(-s) \, ds  \right) \, d t \\
& = - 2 \int_0^{\infty} \frac{d}{dt} \exp \left( -\int_0^t \lambda(-s,-1) \, d s \right) \left( \int_0^t g(-s) \, d s\right) \, d t  \\
& = 2 \int_0^{\infty} \exp \left( -\int_0^t \lambda(-s,-1) \, d s \right)g(-t)  \, d t  \\
& = 2 \exp(U(0)) \int_{-\infty}^{\infty} \exp \left( -U(t) \right) g(t) \, d t - 2 \exp(U(0)) \int_0^{\infty} \exp \left( -U(t) \right) g(t) \, d t \\
& = - \E[Y^+].
\end{align*}

Next,
\begin{align*}
 \E[(Y^+)^2] & = 4 \int_0^{\infty} \lambda(t,+1) \exp \left( -\int_0^t \lambda(s,+1) \ d s \right) \left( \int_0^t g(s) \ d s \right)^2  \ d t \\
 & = 4 \int_0^{\infty} U'(t)  \exp \left( -U(t) \right) \left( \int_0^t g(s) \ d s \right)^2  \ d t
\end{align*}
and similarly
\[ \E[(Y^-)^2] =  4  \int_{-\infty}^0 (-U'(t)) \exp(-U(t)) \left(\int_t^0 g(s) \ d s\right)^2  \ d t. \]
By Assumption~\ref{ass:unimodal} (iii), 
\[ \E[Y^2] = \E[ (Y^+ + Y^-)^2] \leq 2 \E[(Y^+)^2 ] + 2 \E[(Y^-)^2] < \infty.\]
Also by this assumption, 
$\int_0^{T_0^+} g(X(s)) \ ds$ and $\int_{T_{N(t)}^+}^t g(X(s)) \ d s$ are bounded in probability.
Furthermore
\[ \E[\tau^+ + \tau^-] = \int_{-\infty}^{\infty} \exp\left( - U(t) \right) \ d t <\infty \]
since $\pi(t) \propto \exp(-U(t))$ is a probability measure.
By the strong law for renewal processes, \cite[Theorem 1.7.3]{Durrett1996}, $\frac{N(t)}{t} \rightarrow \frac 1 {\E[2 \tau^+ + 2 \tau^-]}$ almost surely.
It follows from Lemma~\ref{lem:CLT-randomtime} (located in the appendix) that
\[ \frac 1 {\sqrt{t}} \sum_{i=1}^{N(t)} Y_i \Rightarrow \mathcal N(0, \E[Y^2]/\E[2 \tau^+ + 2 \tau^-]) \quad \mbox{as $t \rightarrow \infty$}\]
where
\begin{align*} \E[Y^2] & = \E[(Y^+)^2] + \E[(Y^-)^2] - 2 \E[Y^+]^2.
\end{align*}
Combining all terms gives the stated expression for the asymptotic variance.
\end{proof}

\subsection{General approach to the Central Limit Theorem}
\label{sec:abstract-CLT}

The approach of Section~\ref{sec:unimodal} is intuitively appealing. However the required assumptions are very restrictive. In this section we will employ a far more general approach to obtaining a CLT. In particular, this approach allows us to include non-unimodal cases, as well as situations in which the excess switching rate $\gamma$ in~\eqref{eq:lambda-explicit} is non-zero.


First we recall two key results from the literature which will be helpful for our purposes. Recall the definition of a \emph{petite set} from e.g. \cite{MeynTweedie1993-II}.

\begin{assumption}
\label{ass:GlynnMeyn}
$(Z(t))_{t \geq 0}$ is a $\varphi$-irreducible continuous time Markov process in a Borel space $E$ with extended generator $L$. For a function $f : E \rightarrow [1, \infty)$, a petite set $C \in \mathcal B(E)$, a constant $b < \infty$ and a function $V : E \rightarrow [0,\infty)$, $V \in \mathcal D(L)$,
\begin{equation} \label{eq:lyapunov-condition} L V(z) \leq - f(z) + b \1_{C}(z), \quad z \in E.\end{equation}
\end{assumption}

\begin{proposition}[{\cite[Theorem 3.2]{GlynnMeyn1996}}]
\label{prop:solution-Poisson-equation}
Suppose that Assumption~\ref{ass:GlynnMeyn} is satisfied. Then $(Z(t))_{t \geq 0}$ is positive Harris recurrent with invariant probability distribution $\mu$ and $\mu(f) < \infty$. For some $c_0 < \infty$ and any $|g| \leq f$, the Poisson equation
\begin{equation}
 \label{eq:Poisson-equation}
 L \phi = \mu(g) - g
\end{equation}
admits a solution $\phi$ satisfying the bound $|\phi| \leq c_0(V+1)$.
\end{proposition}

Define a sequence of stochastic processes $(Y_n(t))_{t \geq 0}$, $n \in \N$, by
\[ Y_n(t) = \frac 1 {\sqrt{n}} \left( \int_0^{nt} \{ \pi(g) - g(Z(s))\}  \ d s \right), \quad t \geq 0.\]
The following general result establishes sufficient conditions for a functional Central Limit Theorem to hold. Part of the results in this section can be obtained simply by verifying the conditions of the following theorem, although in particular work needs to be done to find suitable functions $f$ and $V$ satisfying Assumption~\ref{ass:GlynnMeyn}.

\begin{proposition}[{\cite[Theorem 4.3]{GlynnMeyn1996}}]
\label{prop:GlynnMeyn-FCLT}
Suppose Assumption~\ref{ass:GlynnMeyn} is satisfied. If $\mu(V^2) < \infty$, then for any $|g| \leq f$ there exists a constant $0 \leq \gamma_g < \infty$ such that under $\P_{\eta}$, $Y_n \Rightarrow \gamma_g B$, with $B$ a standard Brownian motion, as $n \rightarrow \infty$ in $D[0,1]$ for any initial distribution $\eta$.
Furthermore, the constant $\gamma_g^2$ can be defined as $\gamma_g^2 = 2 \int_E \phi(x) \{ \pi(g) - g(x)\} \pi(dx)$, where $\phi$ is the solution to the Poisson equation given in Proposition~\ref{prop:solution-Poisson-equation}.
\end{proposition}

In situations where $\mu(V^2) <\infty$ can not be established, we will have to establish a weaker (non-functional) form of the central limit theorem, which will depend on a CLT for martingales such as~\cite[Theorem 2.1]{Komorowski2012}. We require the following lemmas, the proofs of which may be found in Appendix~\ref{app:technical}.

\begin{lemma}
\label{lem:martingale-expressions}
Suppose Assumption~\ref{ass:GlynnMeyn} is satisfied. Let $g \in \mathcal M(E)$ be measurable, satisfy $|g| \leq f$ and $\mu(g) = 0$. Suppose $\phi$ is a solution to the Poisson equation~\eqref{eq:Poisson-equation} for the generator $L$ given by~\eqref{eq:generator-zigzag} and suppose $\mu(|\phi|) < \infty$.  Define the process
\begin{equation} \label{eq:martingale-definition} M(t) := \phi(Z(t)) - \phi(Z(0)) + \int_0^t  g(Z(s)) \ d s, \quad t \geq 0,\end{equation}
where $(Z(t))_{t \geq 0}$ denote trajectories of the Zig-Zag process.
Then $M$ is a martingale with respect to the stationary measure $\P_{\mu}$. Define $\psi(x) := \half(\phi(x,+1) - \phi(x,-1))$ and for a given trajectory $Z(t) = (X(t), \Theta(t))$ of the Zig-Zag process, let $N(t)$ denote the process counting the switches in $\Theta$, and let $(T_i)_{i=1}^{\infty}$ denote the random times at which these switches occur. 
The quadratic variation process $[M]$ and predictable quadratic variation process $\langle M \rangle$ admit the following expressions:
\begin{align*}
 [M](t) & = 4 \sum_{i=1}^{N(t)} \psi^2(X(T_i)), \quad \mbox{and} \\
 \langle M \rangle(t) & = 4 \int_0^t \lambda(X(s), \Theta(s)) \psi^2(X(s)) \ d s.
\end{align*}
\end{lemma}

\begin{lemma}
\label{lem:expression-psi}
Suppose Assumption~\ref{ass:GlynnMeyn} holds and $\pi(x) V(x,\pm 1) \rightarrow 0$ as $|x| \rightarrow \infty$. Let $g \in \mathcal M(E)$ such that $|g| \leq f$ and $\mu(g) = 0$. Let  $\phi : E \rightarrow \R$ be as in Proposition~\ref{prop:solution-Poisson-equation}. Define $\psi(x) := \half(\phi(x,+1) - \phi(x,-1))$. Then $\psi$ admits the representation~\eqref{eq:expression-psi}.
Furthermore if, for some $\delta \in \R$, we have $\lim_{x \rightarrow \infty} |x|^{\delta} \pi(x) = 0$ and
\begin{equation} \label{eq:condition-psi-asymptotics} \lim_{|x| \rightarrow \infty} \frac{(g(x,+1) + g(x,-1))\pi(x)}{|x|^{\delta} \pi'(x)} = 0,
\end{equation}
then
\[ \lim_{|x| \rightarrow \infty}\frac{ \psi(x)}{|x|^{\delta}} = 0.\]
\end{lemma}


%
%



%



\begin{theorem}[Central Limit Theorem for the Zig-Zag process]
\label{thm:zigzag-clt-general}
Suppose Assumption~\ref{ass:GlynnMeyn} is satisfied for the Zig-Zag process with generator~\eqref{eq:generator-zigzag} and let $g \in \mathcal M(E)$ satisfy $|g| \leq f$ and $\mu(g) = 0$. Furthermore suppose $V$ satisfies $\mu(V) < \infty$, or alternatively $\mu(|\phi|) < \infty$ where $\phi$ is the solution of the Poisson equation given by Proposition~\ref{prop:solution-Poisson-equation}.
Let $\psi$ be given by
\begin{equation}\label{eq:expression-psi}  \psi(x) =\frac 1 { 2 \pi(x)} \int_x^{\infty} \left\{  g(\xi, +1) +  g(\xi,-1) \right\}  \pi(\xi) \ d \xi\end{equation}
and define 
\begin{equation} \label{eq:zigzag-asymptotic-variance} \sigma^2_g := 4 \int_E \lambda(x,\theta) \psi^2(x) \ d \mu(x,\theta).\end{equation} If $\sigma_g^2 < \infty$ then under the stationary distribution $\P_{\mu}$ over the trajectories of the Zig-Zag process,
\[ \frac 1 {\sqrt{t}} \int_0^t g(Z(s)) \ d s \Rightarrow \mathcal N(0, \sigma^2_g) \quad \mbox{as $t \rightarrow \infty$}.\]
\end{theorem}

\begin{proof}
Let $(Z(t))_{t \geq 0} = (X(t),\Theta(t))_{t \geq 0}$ denote the stationary Zig-Zag process defined on an underlying probability space $(\Omega, \mathcal F, (\mathcal F_t), \P_{\mu})$. 
Let $\phi$ denote the solution of the Poisson equation~\eqref{eq:Poisson-equation}, and define the martingale $M$ as in Lemma~\ref{lem:martingale-expressions}, using that $\mu(|\phi|) < \infty$. Indeed, $|\phi| \leq c_0 (V + 1)$ by Proposition~\ref{prop:solution-Poisson-equation}, and it is assumed that either $\mu(V) < \infty$ or else $\mu(|\phi|) < \infty$. By Lemma~\ref{lem:expression-psi}, $\psi(x):= \half (\phi(x,-1) - \phi(x,+1))$ admits the stated expression.
Due to the stationarity of the Zig-Zag process, $M$ is stationary, and $\sigma^2_g := \E |M(1)|^2$.
By~\cite[Theorem 2.1]{Komorowski2012}, it follows that $M(t)/\sqrt{t}$ converges in distribution to $\mathcal N(0,\sigma_g^2)$.
Because $(Z(t))_{t \geq 0}$ is stationary under $\P_{\mu}$, it follows that $\mathcal L(\phi(Z(t)) ) = \mathcal L(\phi(Z(0))) = \mu_{\star}\phi$ (the pushforward of $\mu$ by $\phi$). As a consequence, 
\[ \frac 1 {\sqrt{t}}\left( \phi(Z(t)) - \phi(Z(0))\right) \Rightarrow 0.\] The stated result now follows by combining the obtained limits in~\eqref{eq:martingale-definition}.
\end{proof}

We have now obtained two different expressions for the asymptotic variance, namely \eqref{eq:asvar-alternative} and \eqref{eq:zigzag-asymptotic-variance}. In cases where both Theorem~\ref{thm:unimodal-CLT} and Theorem~\ref{thm:zigzag-clt-general} apply these expression of course have the same value. In Appendix~\ref{sec:equiv_var} we show the equality of both expressions directly.

We will now introduce some specific assumptions on the switching rates which will suffice to establish a CLT for the Zig-Zag process.

\subsubsection{The exponentially ergodic case}

\begin{assumption}
\label{ass:exp-ergodic}
The switching rate $\lambda : E \rightarrow \R$ is continuous and there exists a constant $x_0 > 0$ such that
\begin{itemize}
 \item[(i)] $\inf_{x \geq x_0} \lambda(x,+1) > \sup_{x \geq x_0} \lambda(x,-1)$, and 
 \item[(ii)] $\inf_{x \leq -x_0} \lambda(x,-1) > \sup_{x \leq -x_0} \lambda(x,+1)$.
\end{itemize}
\end{assumption}
In other words, there are constants $M^- > m^- \geq 0$, $M^+ > m^+ \geq 0$, such that
\begin{align*} \lambda(x,+1) & \geq M^+ > m^+ \geq  \lambda(x,-1) \quad & \mbox{for all $x \geq x_0$},  \quad  & \mbox{and}  \\
 \lambda(x,-1) & \geq M^- > m^- \geq \lambda(x,+1) \quad & \mbox{for all $x \leq -x_0$}.
\end{align*}

It is established in \cite[Theorem 5]{bierkensroberts2015} that under these conditions the Zig-Zag process is exponentially ergodic.

\begin{theorem}[CLT and FCLT for the Zig-Zag process in the exponentially ergodic case]
\label{thm:exp-ergodic-clt}
Suppose Assumption~\ref{ass:exp-ergodic} is satisfied. Let $(Z(t))_{t \geq 0}$ denote the Zig-Zag process with generator~\eqref{eq:generator-zigzag}. Then there exists a unique invariant probability distribution $\mu$ on $E$ for $(Z(t))_{t \geq 0}$. Furthermore there are constants $0 < \alpha^+ \leq M^+ - m^+$ and $0 < \alpha^- \leq M^- - m^-$, with $M^{\pm}, m^{\pm}$ as above, such that for any function $g \in \mathcal M(E)$ satisfying $\mu(g) = 0$ and, for $\theta = \pm 1$,
\begin{equation} \label{eq:exp-ergodic-condition-g} \limsup_{x \rightarrow + \infty} \frac 1 {|x|} \log |g(x,\theta)| < \alpha^+ \quad \mbox{and} \quad \limsup_{x \rightarrow -\infty} \frac 1 {|x|} \log |g(x,\theta)| < \alpha^-,\end{equation}
and if $\sigma_g^2$ as given by~\eqref{eq:zigzag-asymptotic-variance} satisfies $\sigma_g^2 < \infty$,
then 
\[ \frac 1 {\sqrt t} \int_0^t g(Z(s)) \ d s  \Rightarrow \mathcal N(0, \sigma_g^2) \quad \mbox{as $t \rightarrow \infty$}.\]

If in addition $\mu( \1_{[0,\infty)}(x) \exp(2 \alpha^+ x)) < \infty$ and $\mu( \1_{(-\infty, 0]}(x) \exp(-2 \alpha^- x)) < \infty$, then $\sigma_g^2 < \infty$ and for any initial distribution $\eta$ on $E$, under $\P_{\eta}$ the process $(Z(t))_{t \geq 0}$ satisfies a Functional Central Limit Theorem, in the sense that 
\[  \left( \frac 1 {\sqrt n} \int_0^{n t} g(Z(s)) \ d s \right)_{t \in [0,1]} \Rightarrow \sigma_g B \quad \mbox{as $n \rightarrow \infty$},\]
where $B$ denotes a standard Brownian motion  and the weak convergence is with respect to the Skorohod topology on $D([0,1])$. 
\end{theorem}

Although the constants $\alpha^{\pm}$ are not explicitly specified in the formulation of Theorem~\ref{thm:exp-ergodic-clt}, their construction can be traced in the proof of \cite[Theorem 5]{bierkensroberts2015}. Note that, irrespective of the value of $\alpha^{\pm}$, ~\eqref{eq:exp-ergodic-condition-g} is satisfied for any sub-exponential function $g$.

\begin{proof}
Assumption~\ref{ass:exp-ergodic} implies Assumption~\ref{ass:well-posedness-stationary-zigzag}. By Proposition~\ref{prop:well-posedness-stationary-zigzag} it follows that $(Z(t))_{t \geq 0}$ admits a unique invariant probability distribution $\mu$. By tracing the proof of \cite[Theorem 5]{bierkensroberts2015}, it follows that there exists a Lyapunov function $V : E \rightarrow [0,\infty)$ such that 
\[ V(x,\theta) = c^+(\theta) \exp(\alpha^+ x), \quad x \geq x_0, \quad \mbox{and} \quad V(x,\theta) =  c^-(\theta) \exp(- \alpha^- x), \quad x \leq - x_0,\]
for some constants $c^{\pm} > 0$ and $\alpha^{\pm}$ as specified in the statement of the theorem, and such that Assumption~\ref{ass:GlynnMeyn} is satisfied with $f := V$. By the stated assumptions on $g$, possibly after a rescaling by a  constant factor, it follows that $|g| \leq f$. By Proposition~\ref{prop:solution-Poisson-equation}, $\mu(f) < \infty$ and there exists a solution $\phi$ for the Poisson equation~\eqref{eq:Poisson-equation} satisfying $\mu(\phi) = 0$ and $|\phi| \leq c_0(V+1)$ for some constant $c_0 > 0$. In particular $\mu(|\phi|) < \infty$. The CLT is therefore a result of Theorem~\ref{thm:zigzag-clt-general}.
Under the stronger assumption, $\mu(V^2) <\infty$ and therefore the FCLT follows by Proposition~\ref{prop:GlynnMeyn-FCLT}.
\end{proof}

\begin{remark}
A sufficient condition for $\sigma_g^2 < \infty$ is that $g \in \mathcal M(E)$ and $\lambda : E \rightarrow [0,\infty)$ are of polynomial growth in $x$. Indeed if $g(x,\theta) = O(|x|^{\beta})$ then by Lemma~\ref{lem:expression-psi}, for any $\delta > \beta$, $\psi(x) = o(|x|^{\delta})$. Then since $\pi(x) = O(\exp(-(M^+- m^+)x))$ for $x \geq x_0$ (and similarly for $x \leq - x_0$), it follows that $\psi^2(x) \lambda(x,\theta) \pi(x)$ has bounded integral.
\end{remark}

\subsubsection{Heavy-tailed distributions}

\begin{assumption}
\label{ass:heavy-tailed}
$\lambda : E \rightarrow [0,\infty)$ is continuous. There exist constants $\alpha > 0$ and $0 \leq \kappa \leq 1$ such that $\lambda(x,+1) \geq \alpha x^{-\kappa}$ for $x > x_0$ and $\lambda(x,-1) \geq \alpha (-x)^{-\kappa}$ for $x <  -x_0$, with $\alpha > 2$ in case $\kappa = 1$. Furthermore $\lambda(x,-1) = 0$ for $x > x_0$ and $\lambda(x,+1) = 0$ for $x < - x_0$.
\end{assumption}



\begin{lemma}
\label{lem:lyapunov-heavy-tailed}
Suppose Assumption~\ref{ass:heavy-tailed} is satisfied. Let $1 \leq \beta < \alpha$ in case $\kappa = 1$, and $1 \leq \beta < \infty$ in case $\kappa < 1$. There exists a norm-like function $V : E \rightarrow [0,\infty)$, and a function $f$ of the form $f(x,\theta) = c |x|^{\beta - 1}$ for some $c > 0$, and $x_1 > 0$ such that
\[ L V(x,\theta) \leq -f(x,\theta), \quad |x| > x_1, \quad \theta \in \{-1,+1\}.\] 
\end{lemma}

\begin{proof}
Let $V$ be given for $x > x_0$ by $V(x,+1) = k x^{\beta}$ and $V(x,-1) = \frac 1 {\beta} x^{\beta}$, with 
\[ 
 k = \begin{cases} \frac{2 \alpha}{\beta(\alpha-\beta)} \quad & \mbox{if $\kappa = 1$,} \\
      \frac {2}{\beta} \quad & \mbox{if $0 \leq \kappa < 1$.}
     \end{cases}
\] 
Then for $x > x_0$,
$LV(x,-1) = - x^{\beta - 1}$ 
and 
\begin{align*} L V(x,+1) & = k \beta x^{\beta -1} + \lambda(x,+1) \left( \frac 1 {\beta} - k \right) x^{\beta} \leq k \beta x^{\beta - 1} + \alpha \left( \frac 1 \beta - k \right)  x^{\beta - \kappa}  \\
 & = \begin{cases} -\frac{\alpha}{\beta} x^{\beta - 1} \quad & \mbox{if $\kappa = 1$,} \\
          2 x^{\beta - 1} - \frac{\alpha}{\beta} x^{\beta - \kappa} \quad & \mbox{if $0 \leq \kappa < 1$.}
        \end{cases}
\end{align*}
In the case $\kappa < 1$, the negative term will dominate for $x$ sufficiently large. It follows in either case that for a suitable constant $c > 0$ and $x_1 \geq x_0$, $LV(x,\pm 1) \leq - c x^{\beta - 1} \leq - 1$ for all $x \geq x_1$. The situation for $x \leq -x_0$ is completely analogous, and within $[-x_0, x_0]$, the function $V$ can be continuously and differentiably extended.
\end{proof}

\begin{remark}
In fact for Lemma~\ref{lem:lyapunov-heavy-tailed} we only require $\alpha > 1$ in case $\kappa = 1$, because this allows us to choose $\beta \in [1, \alpha)$. However in order to obtain $\mu(V) <\infty$ as required for the proof of the following theorem we need the stronger assumption $\alpha > 2$ in case $\kappa = 1$.
\end{remark}



\begin{theorem}[CLT and FCLT for the Zig-Zag process with a heavy-tailed stationary distribution]
\label{thm:heavy-tailed-clt}
Suppose Assumption~\ref{ass:heavy-tailed} is satisfied. 
Let $(Z(t))_{t \geq 0}$ denote the Zig-Zag process with generator~\eqref{eq:generator-zigzag}. Then there exists a unique invariant probability distribution $\mu$ on $E$ for $(Z(t))_{t \geq 0}$. Suppose $g \in \mathcal M(E)$ with $\mu(g) = 0$ and $g(x,\theta) = O(|x|^{\beta - 1})$ where $1 \leq \beta < \alpha - 1$ in case $\kappa = 1$ and $1 \leq \beta < \infty$ in case $\kappa < 1$. Furthermore suppose $\sigma_g^2 := 4 \int_E \lambda(x,\theta) \psi^2(x) \ d \mu(x,\theta) < \infty$, where $\psi$ is given by~\eqref{eq:expression-psi}. 

Then the stationary Zig-Zag process $(Z(t))_{t \geq 0}$ with switching rates $\lambda$ satisfies a CLT with asymptotic variance $\sigma_g^2$, i.e. under the stationary measure $\P_{\mu}$ on the trajectories of the Zig-Zag process,
\[ \frac 1 {\sqrt{t}} \int_0^t g(Z(s)) \ d s \Rightarrow \mathcal N(0,\sigma_g^2) \quad \mbox{as $t \rightarrow \infty$}.\]

If furthermore either 
\begin{itemize}
 \item[(i)] $\kappa < 1$, or
 \item[(ii)] $\kappa = 1$, $\alpha > 3$ and $1 \leq \beta < (\alpha -1)/2$, 
\end{itemize}
then $\sigma_g^2 < \infty$ and for any initial distribution $\eta$ on $E$, under $\P_{\eta}$ the process $(Z(t))_{t \geq 0}$ satisfies a Functional Central Limit Theorem, in the sense that 
\[  \left( \frac 1 {\sqrt n} \int_0^{n t} g(Z(s)) \ d s \right)_{t \in [0,1]} \Rightarrow \sigma_g B \quad \mbox{as $n \rightarrow \infty$},\]
where $B$ denotes a standard Brownian motion  and the weak convergence is with respect to the Skorohod topology on $D([0,1])$. 
\end{theorem}

\begin{proof}
Assumption~\ref{ass:heavy-tailed} implies Assumption~\ref{ass:well-posedness-stationary-zigzag} so that by Proposition~\ref{prop:solution-Poisson-equation} there is a unique invariant probability distribution $\mu$.
If $\kappa = 1$ in Assumption~\ref{ass:heavy-tailed} then $\frac{d \mu}{dx}(x,\theta) =O( |x_0/x|^{\alpha})$. Because $\alpha > 2$ we can choose $1 \leq \beta < \alpha - 1$ in Lemma~\ref{lem:lyapunov-heavy-tailed}, and it follows that the Lyapunov function $V(x,\theta) = O(|x|^{\beta})$ satisfies $\mu(V) < \infty$.
If $0 \leq \kappa < 1$ then $\frac{d \mu}{dx}(x,\theta) = O(\exp( -\alpha/(1-\kappa) |x|^{1-\kappa}))$ and again $\mu(V) < \infty$. The CLT now follows from Theorem~\ref{thm:zigzag-clt-general}.
Under the stronger assumptions, $\mu(V^2) < \infty$ using the above asymptotic analysis,  so that the FCLT follows from Proposition~\ref{prop:GlynnMeyn-FCLT}.
\end{proof}

\begin{remark}
A sufficient condition for $\sigma_g^2 < \infty$ in case $\kappa = 1$ is that $\alpha > 2$, $1 \leq \beta < \min(\alpha - 1,\half \alpha)$ and $\lambda(x,+1) = O(x^{-1})$. Indeed, in this case there exists a $\delta \in (\beta, \alpha/2)$. Since $\pi(x) = O(|x|^{-\alpha})$ and $\delta < \alpha$ we have that $\pi(x) |x|^{\delta} \rightarrow 0$. Furthermore~\eqref{eq:condition-psi-asymptotics} is satisfied as $g(x) = O(|x|^{\beta - 1})$ and $\pi(x)/\pi'(x) = O(|x|^{-1})$, so we may deduce from Lemma~\ref{lem:expression-psi} that $\psi(x) = o(|x|^{\delta})$. Hence $\lambda(x) \psi^2(x) \pi(x) = o(|x|^{2 \delta -1 - \alpha}) = o(|x|^{-1})$ using that $\delta < \alpha/2$.
\end{remark}

\subsubsection{Comparison with Langevin diffusion}

Let $A$ denote the generator of the Langevin diffusion with invariant density $\pi(x) = \exp(-U(x))/k$, i.e.
\[ A f(x) = f''(x) - U'(x) f'(x),\]
with domain including at least all twice continuously differentiable functions $f$ for which $A f$ is a bounded continuous function.

\begin{proposition}
\label{prop:langevin-asvar}
Suppose $g \in L^2(\pi)$ with $\pi(g) = 0$ and let $\psi$ as in~\eqref{eq:expression-psi}. If $\psi \in L^2(\pi)$ then under the stationary measure $\P_{\pi}$ the Langevin diffusion $(X(t))_{t \geq 0}$ with generator $A$ satisfies the CLT with asymptotic variance is given by $\widetilde \sigma_g^2 := 2 \int_{\R} |\psi(x)|^2 \pi(x) \ d x < \infty$, i.e.
\[ \lim_{t \rightarrow \infty} \frac 1 {\sqrt t} \int_0^t g(X_s) \ d s \Rightarrow \mathcal N(0, \widetilde \sigma_g^2).\]
Conversely, if $\int_{\R} |\psi(x)|^2 \pi(x) \ d x = \infty$, then $\limsup_{t \rightarrow \infty} \frac 1 t \operatorname{Var}_{\pi} \left( \int_0^t g(X_s)  \ ds \right) = \infty$.
\end{proposition}

The proof of this result may be found in Appendix~\ref{app:langevin-asvar}.

%
%
%

In cases where both a CLT holds for the Langevin diffusion and the Zig-Zag process, and the function of interest $g$ does not depend on $\theta$, we can compare the asymptotic variances, given by \begin{align*}
\widetilde \sigma_g^2 & = 2 \int_{\R} \psi^2(x) \pi(x) \ d x \quad &  \mbox{(Langevin asymptotic variance)}, \\
\sigma_g^2 & = 2 \int_{\R} \left( \lambda(x,+1) + \lambda(x,-1) \right) \psi^2(x) \pi(x) \ d x \\ & = 2 \int_{\R} \left( |U'(x)| + 2 \gamma(x) \right) \psi^2(x) \pi(x) \ d x  \quad & \mbox{(Zig-Zag asymptotic variance)}.
\end{align*}
where we used~\eqref{eq:lambda-explicit} to obtain the last equality.

Trivially, if $\lambda(x,+1) + \lambda(x,-1) \leq 1$ for all $x \in \R$, the asymptotic variance of the Zig-Zag process is less than or equal to the asymptotic variance of the Langevin diffusion, but this is a very restrictive condition. More generally, the asymptotic variance of the Zig-Zag process is smaller than that of the Langevin if the switching rates are small where $\psi^2 \pi$ has most of its mass. It is also clear from the above expression that having a non-zero excess switching rate $\gamma$ increases the asymptotic variance of the Zig-Zag process.

%
%
%

\subsection{Examples}
\label{sec:clt_examples}


To illustrate the effectiveness of the developed theory we consider several examples. We consider (i)  Gaussian distributions, which have light tails and for which the associated Zig-Zag process is exponentially ergodic, and (ii) Student t-distributions, which are heavy tailed so that the associated Zig-Zag process is not exponentially ergodic. For both families of distributions we will consider two types of observables: (a) moments  and (b) tail probabilities.

\subsubsection{Gaussian distribution}

The family of centered one-dimensional Gaussian distributions $\mathcal N(0, \nu^2)$ is described by the potential functions and canonical switching rates
\begin{equation} \label{eq:gaussian-setup} U(x) = \frac {x^2}{2 \nu^2} \quad  \mbox{and} \quad \lambda(x,\theta) = \left( \theta x / \nu^2 \right)^+.\end{equation}

\begin{example}[Moments of a Gaussian distribution]
\label{ex:gaussian-moments}
First we consider the asymptotic variance associated with the $k$-th moment for positive integer values of $k$. This corresponds to the mean-zero functional $g(x) = x^k - m_k$, where 
\[ m_k = \frac 1 {\sqrt{ 2 \pi \nu^2}} \int_{\R} x^k \exp(-x^2/2\nu^2) \ d x  = \begin{cases} 0 \quad & \mbox{if $k$ is odd}, \\
           \nu^k (k-1)!! = \frac{\nu^{k} k!}{2^{k/2} (k/2)!}  \quad &\mbox{if $k$ is even}.
          \end{cases}\]
Assumption~\ref{ass:unimodal} is satisfied for any $k \geq 0$ so that a CLT holds by Theorem~\ref{thm:unimodal-CLT}. The asymptotic variance can be computed using~\eqref{eq:asvar-alternative} to be
\[ \sigma_g^2 = \frac{\nu^{2 k + 1}}{\sqrt{ 2 \pi}}  \times \begin{cases} 2^{k+2} \left( \frac{ 2 (k!)}{k+1} - \frac{\left( (\frac{k-1}{2})!\right)^2}{2} \right)
\quad & \mbox{for $k$ odd,} \\
 \frac{8 (k!)^2}{2^k ((k/2)!)^2} + \frac{8 (k!) (2^k -k- 2)}{k+1} \quad & \mbox{for $k$ even.}  \end{cases}\]
 The variance of $g$ under $\pi$ is given by
\[ \Var_{\pi}(g) = \nu^{2k} \times \begin{cases}  (2k-1)!! \quad &\mbox{for $k$ odd}, \\
                    (2k-1)!! -  ((k-1)!!)^2 \quad &\mbox{for $k$ even}.
                   \end{cases}
 \]
In order to compare the asymptotic variance of the Langevin diffusion, we compute 
\begin{align*} \psi(x) & = \exp(x^2/2\nu^2) \int_{x}^{\infty} (\xi^k - m_k) \exp(-\xi^2/2\nu^2) \ d \xi. 
\end{align*}
Expressions for $\psi(x)$ for different values of $k$ are given, along with the computed asymptotic variance for the Zig-Zag process ($\sigma_g^2$) and Langevin diffusion ($\widetilde \sigma_g^2$), in the following table. 

{\begin{center}
 \begin{tabular}{c|c|c|c|c}
  $k$ & 1 & 2 & 3 & 4   \\
  \hline
  $\Var_{\pi}(g)$ & $\nu^2$ & $2 \nu^4$ & $15 \nu^6$ & $96 \nu^8$  \\
  $\psi(x)$ & $\nu^2$ & $\nu^2 x$ & $\nu ^2 \left(x^2 + 2 \nu ^2 \right)$ & $\nu ^2 x \left(x^2 + 3 \nu ^2\right)$   \\  
  $\sigma_g^2$ & $2 \sqrt{\frac{2}{\pi }} \nu ^3$ & $4 \sqrt{\frac{2}{\pi }} \nu ^5$ & $40 \sqrt{\frac{2}{\pi }} \nu ^7$ & $228 \sqrt{\frac{2}{\pi }} \nu ^9$  \\
  $\widetilde \sigma_g^2$ & $\nu^4$ & $\nu ^6$ & $11 \nu ^8$ & $42 \nu^{10}$
 \end{tabular}
 \end{center}
}

For each of these moments we note that $\sigma^2_g/\widetilde{\sigma}^2_g \propto \nu^{-1}$,  which suggests that for large variance distributions, the variance of an estimator for $\pi(g)$ using the Zig-Zag process will be considerably lower than that of an estimator generated from a Langevin trajectory.

The result of Theorem~\ref{thm:unimodal-CLT} can be strengthened since by Theorem~\ref{thm:exp-ergodic-clt} the Functional Central Limit Theorem holds for this entire family of examples.
\end{example}

\begin{example}[Tail probabilities for a Gaussian distribution]
Next consider the tail probabilities $p_a := \pi(x \geq a)$ for a $\mathcal N(0, \nu^2)$-distribution. The potential and associated switching rates are given by~\eqref{eq:gaussian-setup}. We have $p_a = 1 - \Phi(a/\nu)$ and $g(x) = \1_{[a,\infty)}(x) - p_a$.
Assumption~\ref{ass:unimodal} is satisfied for any value of $\nu > 0$ so that Theorem~\ref{thm:unimodal-CLT} gives a CLT. Again, using Theorem~\ref{thm:exp-ergodic-clt} we obtain a functional CLT in this family of examples.
Computing the necessary integrals in~\eqref{eq:asvar-alternative} gives the asymptotic variance
\begin{equation}
\label{eq:tailprob_gaussian}
\sigma_g^2 = \frac{-4 a(1-p_a) p_a  \nu \sqrt{2 \pi} + 4(1-2p_a) \nu^2 \exp \left(-a^2/(2 \nu^2) \right) + (8 - 2 \pi) p_a^2 \nu^2}{\sqrt{2 \pi \nu^2}},
\end{equation}
while the variance of $g$ is given by $\Var_{\pi}(g) = p_a(1-p_a)$.
%


In Figure \ref{fig:gaussian_tailprob} we compare the expression~\eqref{eq:tailprob_gaussian} with the variance estimated from $10^5$ independent simulations of the Zig-Zag process, for different values of $\nu^2$.  

\begin{figure}[!ht]
\begin{subfigure}[b]{0.45 \textwidth}
 \includegraphics[width=\textwidth]{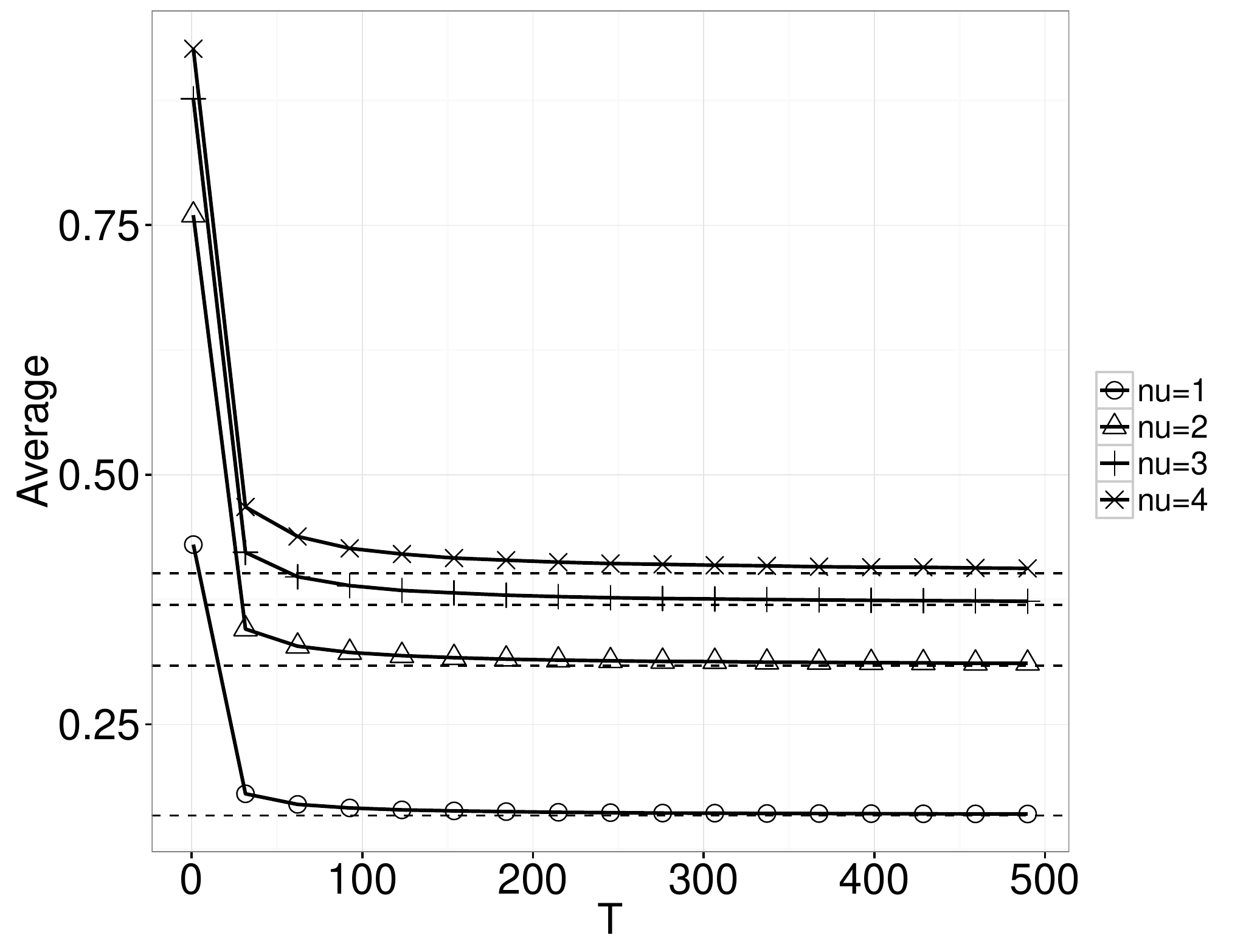}
 \caption{Plot of ergodic average $\pi_T(f)$ of $f(x)=\1_{[1,\infty)}(x)$ as a function of time $T$, for the 1D Zig-Zag process ergodic with respect to $\mathcal{N}(0,\nu^2)$ for different values of  $\nu$. }
 \label{fig:gaussian_tail_prob_means}
\end{subfigure}
    \hfill
\begin{subfigure}[b]{0.5 \textwidth}
 \includegraphics[width=\textwidth]{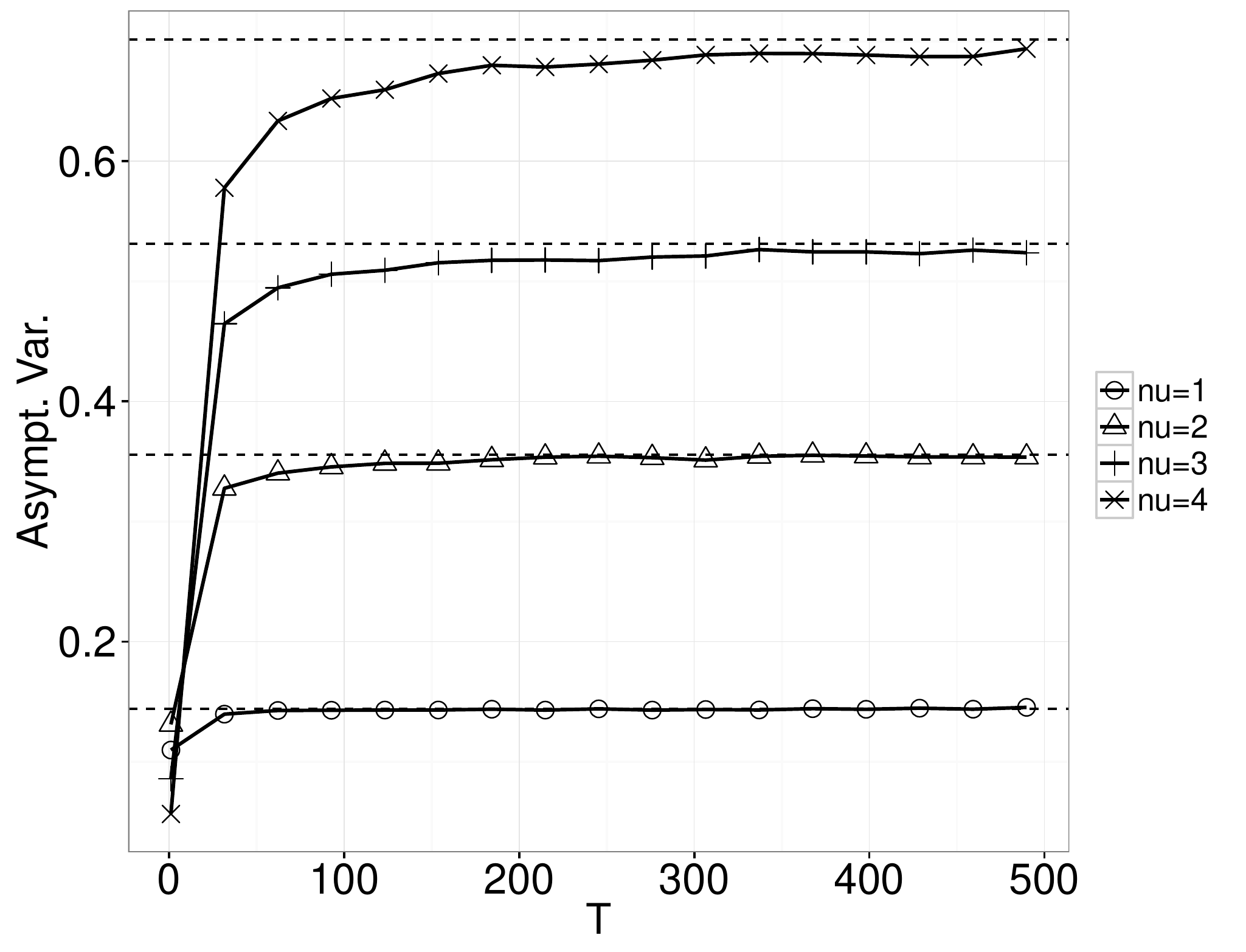}
 \caption{Corresponding values of $T\mbox{Var}[\pi_T(f)]$ as a function of time $T$, for different values of $\nu$.  The dashed lines denote the corresponding asymptotic variance obtained via~\eqref{eq:tailprob_gaussian}.  }
 \label{fig:gaussian_tail_prob_vars}
\end{subfigure}
 \caption{Mean and Variance estimates for the tail probabilities of a Gaussian $\mathcal{N}(0, \nu^2)$ distribution obtained from simluations compared to predicted estimates.\label{fig:gaussian_tailprob}}
\end{figure}

\end{example}




\subsubsection{Student t-distribution}

Consider the family of Student-t distributions with $\nu > 0$ degrees of freedom,
\begin{equation} \label{eq:heavy-tailed-example} \pi(x) \propto \left(1 + \frac{x^2}{\nu} \right)^{-\frac{\nu+1}{2}},\end{equation}  and let $\lambda$ denote the canonical switching rates, given by
\begin{equation}
\label{eq:student_switching}
 \lambda(x,\theta) = \begin{cases} \frac{(\nu + 1)|x|}{\nu + x^2} \quad & \mbox{if $\theta x > 0$}, \\
                        0 \quad & \mbox{if $\theta x \leq 0$}.
                       \end{cases}
\end{equation}

\begin{example}[Moments for the Student t-distribution]
\label{ex:student-t-moments}
For integer values of $k$ with $0 \leq k < \nu$ we can compute the values of the moments to be
\[ m_k := \int_{\R} x^k \pi(x) \ d x = \begin{cases} 0 \quad & \mbox{if $k$ is odd,} \\
                                        \frac{\nu^{(k+1)/2} \Gamma \left( \frac{k+1}{2} \right) \Gamma \left( \frac{\nu-k}{2} \right)}{\sqrt{\pi \nu} \Gamma\left( \frac{\nu}{2} \right)} \quad & \mbox{if $k$ is even.}
                                       \end{cases}\]
The mean-zero function representing the observable of interest is $g(x) = x^k - m_k$. Assumption~\ref{ass:unimodal} is satisfied if $k < (\nu-1)/2$. Moreover we may apply Theorem~\ref{thm:heavy-tailed-clt} with $\alpha < \nu+1$, $\gamma = 1$ and $\beta = k + 1$ to see that in the above cases a functional CLT is satisfied under the stated assumption that $k < (\nu - 1)/2$. 

This may be compared to the Random Walk Metropolis algorithm. In \cite[p. 796]{JarnerRoberts2007} it is established that for a finite variance proposal distribution, the range of parameter values for which a CLT holds is $k < \nu/2 - 1$ which is slightly more restrictive. By tuning the proposal distribution in RWM to have the same decay in the tails, this range can be improved to $k < \nu/2$.

%

Using~\eqref{eq:asvar-alternative} we obtain, for the Zig-Zag process,
\[ \sigma_g^2 = \frac{2  \nu^{k+1} \Gamma(k+1) \Gamma \left( \frac{ \nu-2k - 1}{2} \right)}{(1+k)\sqrt{\pi \nu}}- \frac{\nu^{k+1} \Gamma \left( \frac {1+k}{2} \right)^2 \Gamma \left( \frac{\nu-k}{2} \right)^2}{\sqrt{\pi \nu} \Gamma\left( \frac{ \nu}{2} \right) \Gamma \left(\frac{\nu+1}{2} \right)}  \quad  \mbox{if $k$ is odd.} \]
For $k$ even an also explicit but more cumbersome expression can be obtained from~\eqref{eq:asvar-alternative}.

It may be verified that $\psi(x) = O(|x|^{k+1})$, $\lambda(x,\theta) = O(|x|^{-1})$ and $\pi(x) = O(|x|^{-(\nu+1)}$ as $|x| \rightarrow \infty$. In particular the Langevin asymptotic variance, $\widetilde \sigma_g^2 = 2 \int_{\R} \psi^2(x) \pi(x) \ d x$ is finite if and only if $k < (\nu - 2)/2$, so that the Zig-Zag process has finite asymptotic variance for a wider range of combinations of $k$ and $\nu$.
\end{example}

\begin{example}[Tail probabilities for the Student t-distribution]
\label{ex:student-t-tails}
Suppose now we wish to consider the behaviour of the Zig-Zag process with respect to the observable given by the tail probability $p_a = \int_a^{\infty} \pi(x) \ d x$ for $a \in \R$. The associated functional of interest is $g(x) = \1_{[a, \infty)}(x) - p_a$. 
Assume $a \geq 0$ for simplicity. Assumption~\ref{ass:unimodal} is satisfied if $\nu > 1$, so that for these values of $\nu$ a CLT holds. Using Theorem~\ref{thm:heavy-tailed-clt} a functional CLT holds at least for those cases for which $\nu > 2$.

It may be verified that $\psi(x) = O(|x|)$, $\lambda(x,\theta) = O(|x|^{-1})$, and $\pi(x) = O(|x|^{-(\nu+1)})$ as $|x| \rightarrow \infty$. Using Proposition~\ref{prop:langevin-asvar} the asymptotic variance $\widetilde \sigma_g^2 = 2 \int_{\R} \psi^2(x) \pi(x) \ d x$ of the associated Langevin diffusion is finite if and only if $\nu > 2$. So for heavy tailed distributions the Zig-Zag process allows for a larger range of parameter values $\nu$ with finite asymptotic variance.

After evaluating the necessary integrals in~\eqref{eq:asvar-alternative}, we find the asymptotic variance of the Zig-Zag process to be
\begin{equation}
\label{eq:student_asympt_var}
\begin{aligned}
 \sigma_g^2 
 &  =  \frac {4(1-2 p_a) \nu}{z (\nu-1)} \left( 1 + \frac {a^2}{\nu} \right)^{(1-\nu)/2} - 4 a(1-p_a)p_a  +  \frac{8 p_a^2 \nu}{z(\nu-1)}- z p_a^2 
\end{aligned}
\end{equation}
where
\[ z = \int_{\R} \exp(-U(x)) \ d x =  \frac{\sqrt{\nu \pi} \Gamma(\nu/2) }{\Gamma((\nu+1)/2)} \]
and, writing ${}_2F_1$ for the hypergeometric function,
\[ p_a = \frac 1 {z} \int_a^{\infty} \exp(-U(x)) \ d x = \frac 1 2 - \frac{a \Gamma((\nu+1)/2) {}_2 F_1 \left( \frac 1 2 ,\frac{\nu+1}{2}, \frac{3}{2}, -\frac{a^2}{\nu} \right)}{\sqrt{ \pi \nu} \Gamma(\nu/2)}.\]

For $\nu = 2$, the above expressions simplify to
\[ \sigma_g^2 = \frac{ \sqrt{2}  + 2 a + \sqrt{2} a^2 - a\sqrt{4 + 2 a^2}}{2 + a^2} \quad \mbox{and} \quad \Var_{\pi}(g) = p_a(1-p_a) = \frac 1 { 4+2a^2}, \]
whereas for other values of $\nu$ the expression for the asymptotic variance can typically not be significantly simplified. 
See Figure~\ref{fig:heavy-tailed-experiment} for an experimental verification of these results.  We see good agreement with theoretical predictions. Also from Figure~\ref{fig:tail_prob_vars} the rescaled variance of the estimator for $\nu = 1$ appears to diverge to infinity as $T\rightarrow \infty$, which suggests that no CLT holds in this case, and thus the condition $\nu > 1$ is indeed tight.


\begin{figure}[!ht]
\begin{subfigure}[b]{0.45 \textwidth}
 \includegraphics[width=\textwidth]{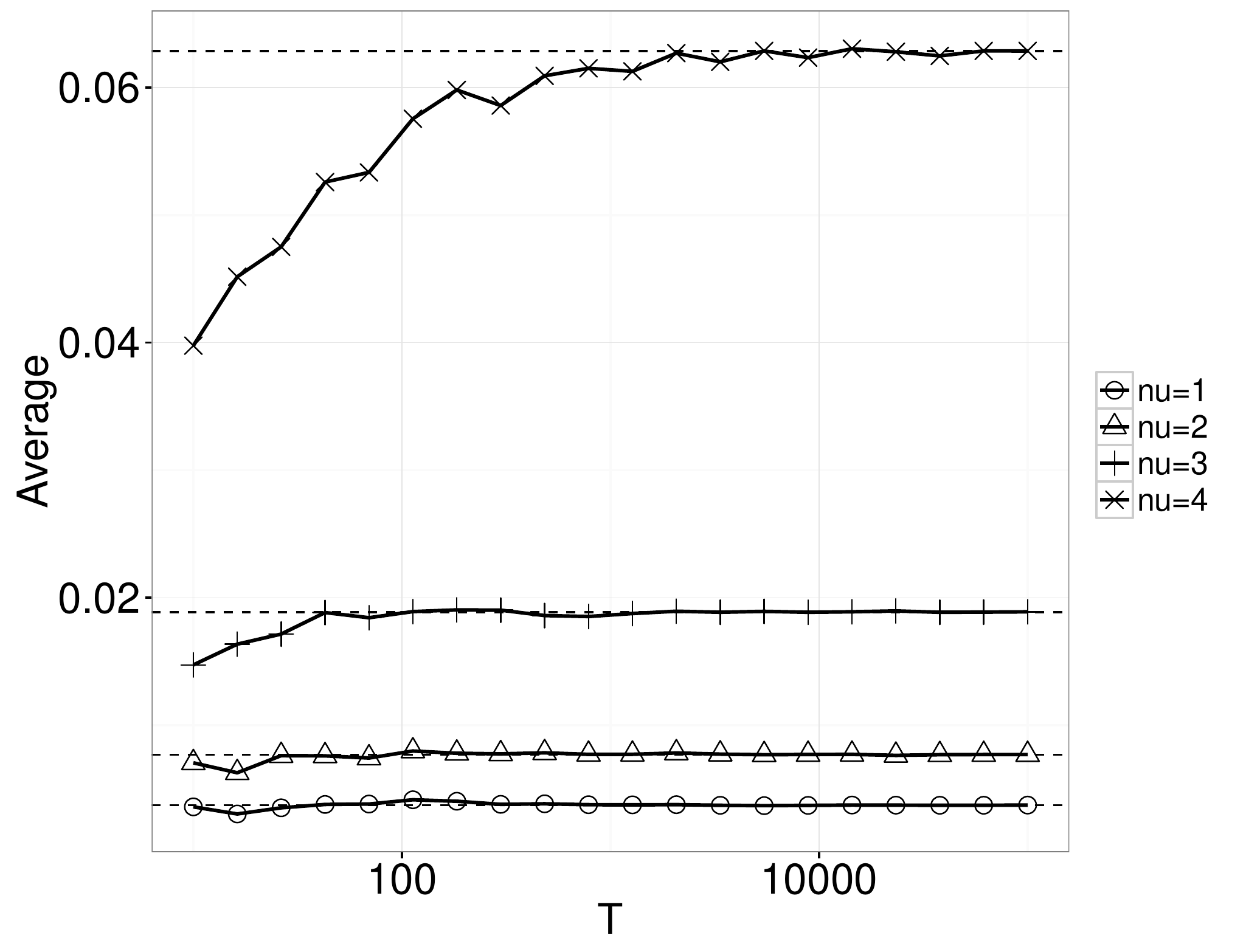}
 \caption{Plot of ergodic average $\pi_T(f) =\frac 1 T \int_0^T f(X(s))\,ds$ of $f(x)= \1_{[5,\infty)}(x)$ as a function of time $T$, for the 1D Zig-Zag process ergodic with respect to~\eqref{eq:heavy-tailed-example} for different values of $\nu$. }
 \label{fig:tail_prob_means}
\end{subfigure}
    \hfill
\begin{subfigure}[b]{0.45 \textwidth}
 \includegraphics[width=\textwidth]{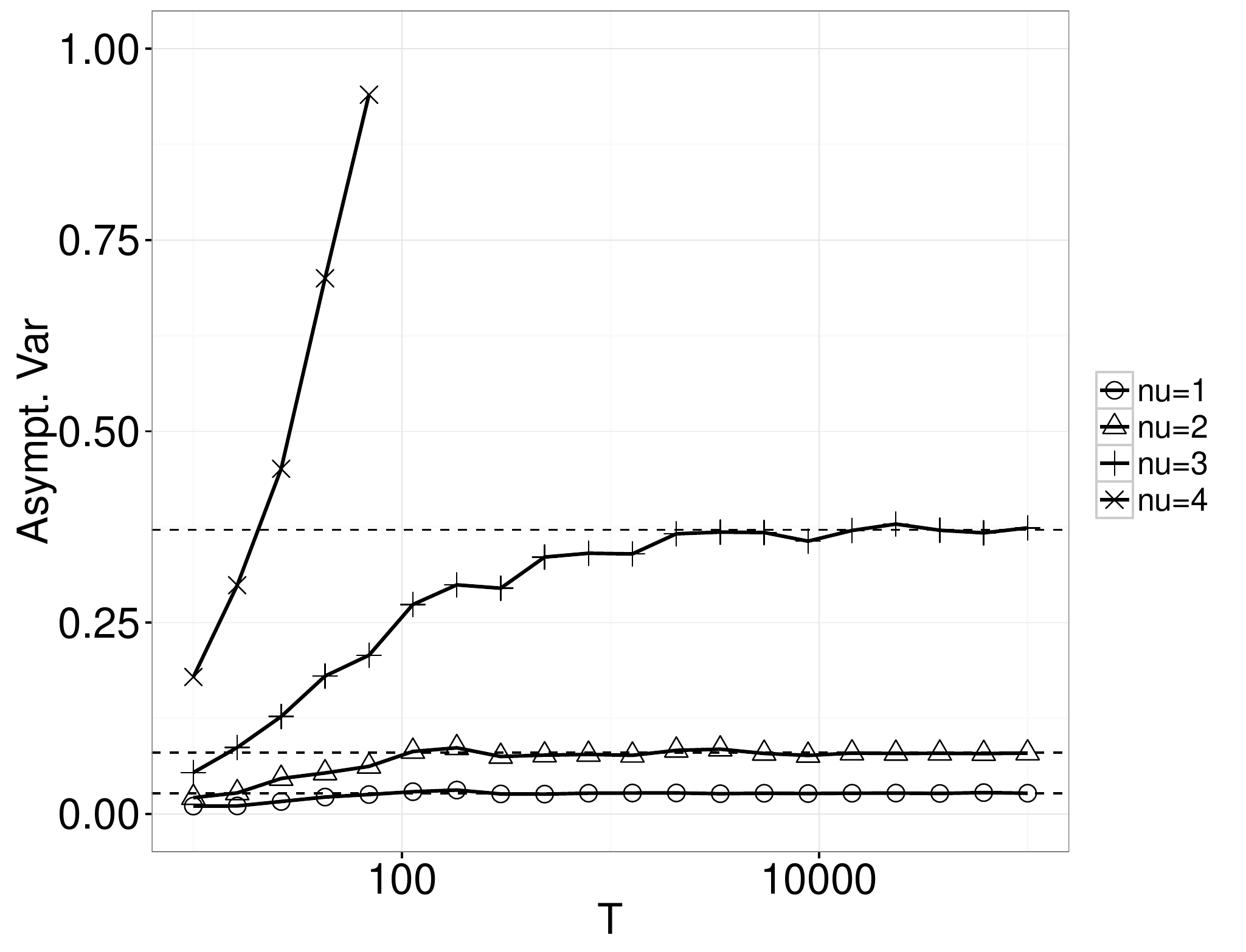}
 \caption{Corresponding values of $T\mbox{Var}[\pi_T(f)]$ as a function of time $T$, for different values of $\nu$, approximated from $10^4$ independent realisations of the Zig-Zag process.}
 \label{fig:tail_prob_vars}
\end{subfigure}

\caption{Convergence towards the ergodic average and asymptotic variance of the Zig-Zag process corresponding to the family of Student-t distributions with $\nu$ degrees of freedom and with the tail probabilities $\1_{[5,\infty)}$ as observable. The dashed lines denote the analytically derived values.}
\label{fig:heavy-tailed-experiment}
  \end{figure}

\end{example}

\section{Diffusion limit of the Zig-Zag process}
\label{sec:diffusive}
In this section we will consider the one dimensional Zig-Zag process with switching rates of the form
\[
\lambda(x,\theta)=\max(0,\theta U'(x))+\gamma(x),
\]
for a general non-vanishing space-dependent switching rate $\gamma$. 
An example arising from applications where $\gamma$ is positive is when Zig-Zag sampling is used in combination with sub-sampling, as discussed in \cite{BierkensFearnheadRoberts2016}. It is observed in simulations that this gives rise to diffusive behaviour. In this section we show that under an appropriate time change the Zig-Zag process converges weakly to an It\^{o} diffusion, ergodic with respect to $\pi$, with space dependent diffusion coefficient inversely proportional to the switching rate $\gamma$.

We shall focus
on behaviour of the Zig-Zag process in the large $\|\gamma\|_{\infty}$
limit. To this end, we shall introduce the rescaling $\gamma^{\epsilon}=\epsilon^{-1}\gamma,$
and denote by $Z^{\epsilon}(t)=(X^\epsilon(t),\Theta^\epsilon(t))$ the corresponding Zig-Zag process,
with generator defined by
\[
L^{\epsilon}f(x,\theta)=\theta \partial_x f(x,\theta)+\left(\lambda^{0}(x,\theta)+\gamma^{\epsilon}(x)\right)\left[f(x,-\theta)-f(x,\theta)\right],
\]
where $\lambda^{0}(x,\theta)=\max(0,\theta U'(x)).$ Our objective
is to prove the following result.
\begin{theorem}
\label{thm:diffusive}
Suppose that $\gamma\in C^{1}(\mathbb{R})$ is positive. Consider the
process $Z^\epsilon(t) =(X^{\epsilon}(t),\Theta^\epsilon(t))$ with initial condition $(X^\epsilon(0),\Theta^\epsilon(0)) \sim \eta$ on $E$.  Suppose that the It\^{o} SDE 
\begin{equation}
\label{eq:limiting_sde}
d\xi(t)=-\frac{1}{2}\left(\frac{U'(\xi(t))}{\gamma(\xi(t))}\,+\frac{\gamma'(\xi(t))}{\gamma(\xi(t))^{2}}\right)\,dt+\sqrt{\frac{1}{\gamma(\xi(t))}}\,dW(t),
\end{equation}
where $\xi(0)$ is distributed according to the marginal distribution of $\eta$ with respect to $x$, and where $W(t)$ is a standard Brownian motion independent from $\xi(0),$ has a unique weak solution for $t\geq 0$.   Then as $\epsilon\rightarrow 0$, the process $X^\epsilon(t/\epsilon)$  converges weakly in $C([0,\infty),\mathbb{\mathbb{R}})$ to the solution $\xi(t)$ of~\eqref{eq:limiting_sde}.
\end{theorem}

\begin{remark} If the process $(\xi(t))_{t \geq 0}$ exists and is non-explosive, then it is ergodic with unique stationary  distribution $\pi(x)\propto\exp(-U(x))$.
\end{remark}
\noindent
To prove this result, we will follow an approach similar to that of \cite[Theorem 1.5]{fontbona2016long}.  The main distinction is that, in \cite[Theorem 1.5]{fontbona2016long} the authors introduce a random time-change for the PDMP which produces a limiting SDE with additive noise.  On the other hand, the limiting SDE~\eqref{eq:limiting_sde} is qualitatively different,  in particular it will have multiplicative noise dependent on the switching rate $\gamma$ and moreover is ergodic with respect to the unique stationary disitribution $\pi$.   The proof of Theorem \ref{thm:diffusive} will be deferred to Section \ref{sec:proof_diffusive}.

 \begin{example}
We demonstrate the conclusions of Theorem \ref{thm:diffusive} using a simple example.   Given $U(x) = x^2/(2 \sigma^2)$ consider the family of Zig-Zag processes $Z^\epsilon(t) = (X^\epsilon(t), \Theta^\epsilon(t))$ with switching rates
\begin{equation}
\label{eq:switching_example}
	\lambda^\epsilon(x, \theta) = \max(0, \theta U'(x)) + \frac{1}{\epsilon}\gamma(x),
\end{equation}
where we choose $\gamma(x) = (1 + x^2)$ for a positive parameter $\epsilon>0$.   The resulting process is ergodic, with unique invariant distribution $\pi \sim \mathcal{N}(0, \sigma^2)$.  Applying Theorem \ref{thm:diffusive}  we know that, in the limit $\epsilon\rightarrow 0$, the time-changed process $X^{\epsilon}(t/\epsilon)$ will converge weakly to an It\^{o} diffusion process $\xi(t)$ given by the unique solution of
\begin{equation}
\label{eq:langevin_example}
	d\xi(t) = -\left(\frac{1}{2\sigma^2}\frac{\xi(t)}{1+\xi^2(t)}\, + \frac{2\xi(t)}{(1+\xi^2(t))^2}\right)\,dt + \sqrt{\frac{1}{1+\xi^2(t)}}\,dW(t).
\end{equation}
It is straightforward to show that $(\xi(t))_{t \geq 0}$ is an ergodic process with unique invariant distribution $\pi$.  In Figure \ref{fig:zig_zag_diffusion} we demonstrate this result numerically.  Choosing $\sigma^2=1$ and for $\epsilon=10, 1, 0.1$ we plot a histogram  of  the values of $Z^\epsilon(t)$ at values $t/\epsilon=1$,$10$, $20$ and $50$ over $10^4$ independent realisations starting from $X^\epsilon(0)=2.0$.   We compare the result with the corresponding distribution of the diffusion process~\eqref{eq:langevin_example} denoted by the solid line.    While for larger values of $\epsilon$ there is a clear discrepancy between $X^\epsilon(t)$ and $\xi(t)$,  as the speed of the switching rate increases, the Zig-Zag process displays increasing random walk behaviour and shows very good agreement with the diffusion process.  

\begin{figure}[!ht]
\begin{subfigure}[b]{0.5 \textwidth}
 \includegraphics[width=\textwidth]{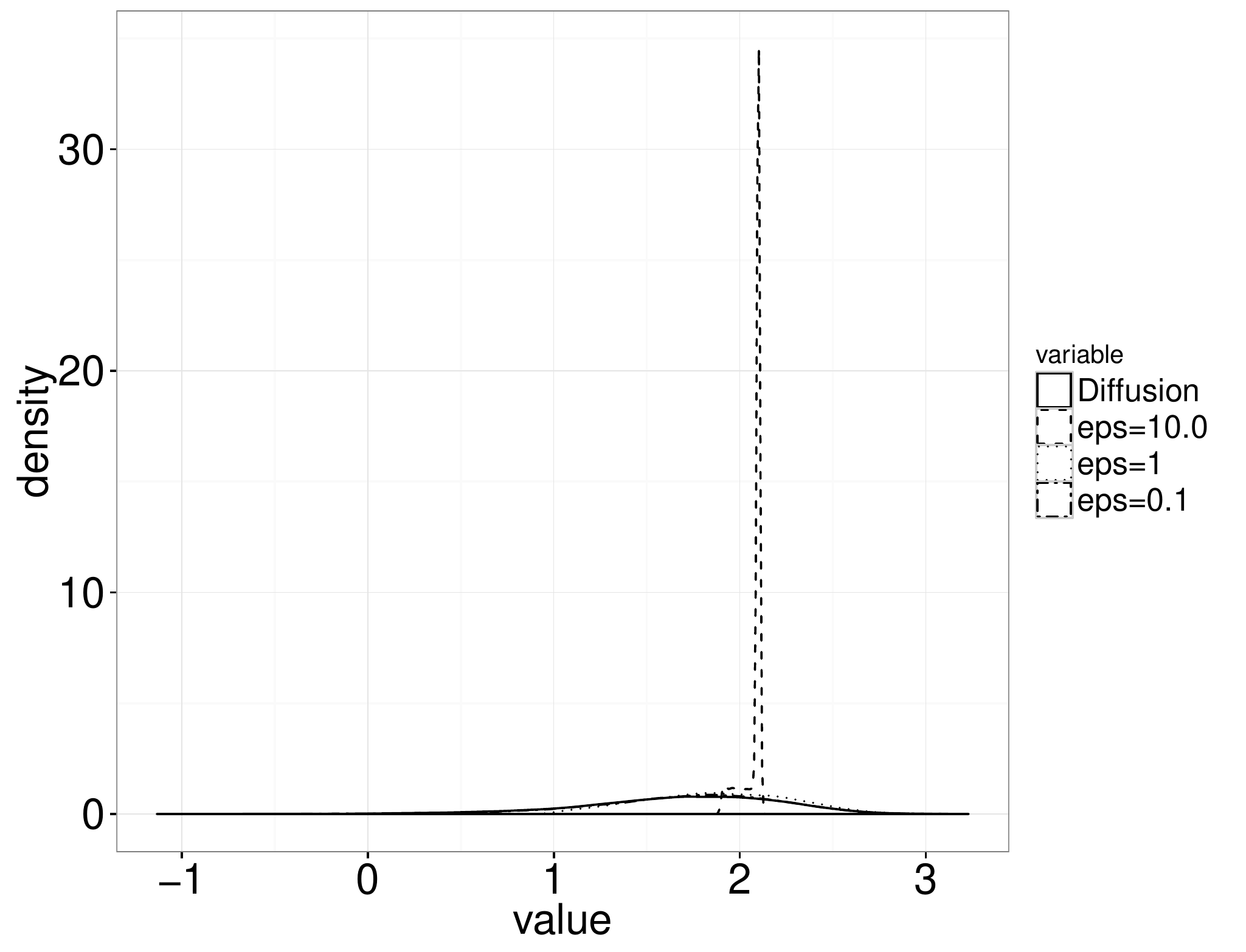}
 \caption{$T=1$}
\end{subfigure}
    \hfill
\begin{subfigure}[b]{0.5 \textwidth}
 \includegraphics[width=\textwidth]{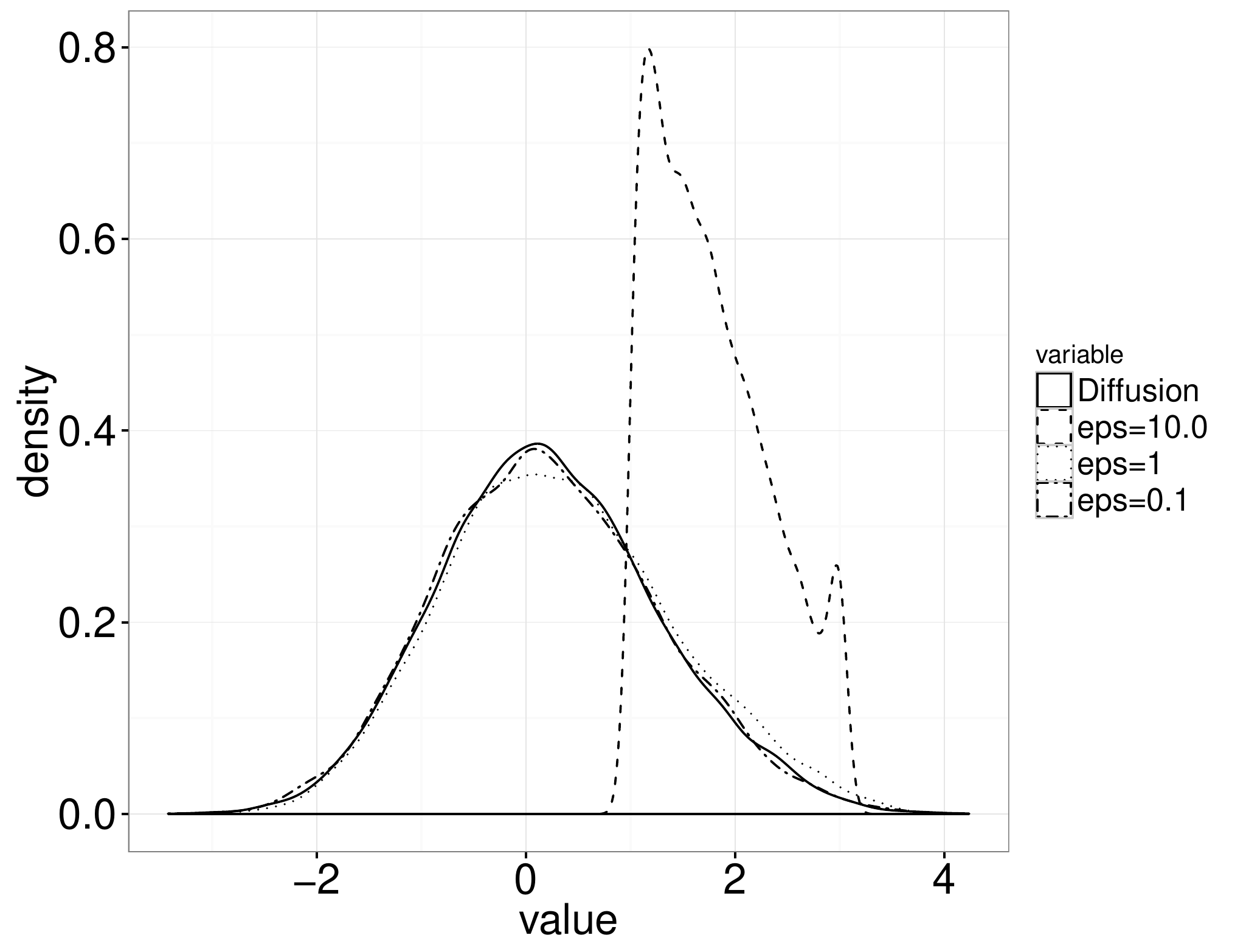}
 \caption{$T=10$}
\end{subfigure}
    \\
\begin{subfigure}[b]{0.5 \textwidth}
 \includegraphics[width=\textwidth]{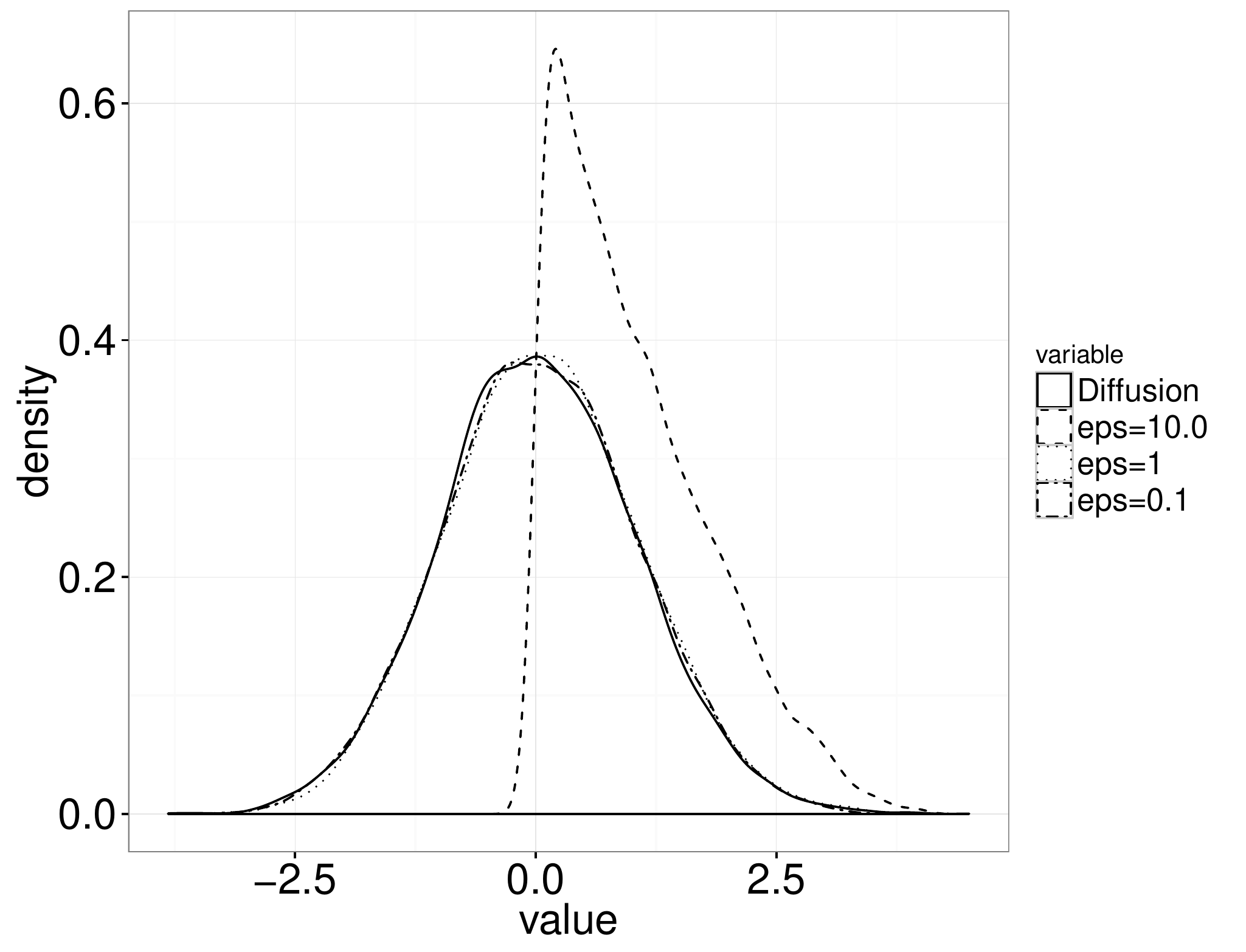}
 \caption{$T=20$}
\end{subfigure}  
\hfill  
\begin{subfigure}[b]{0.5 \textwidth}
 \includegraphics[width=\textwidth]{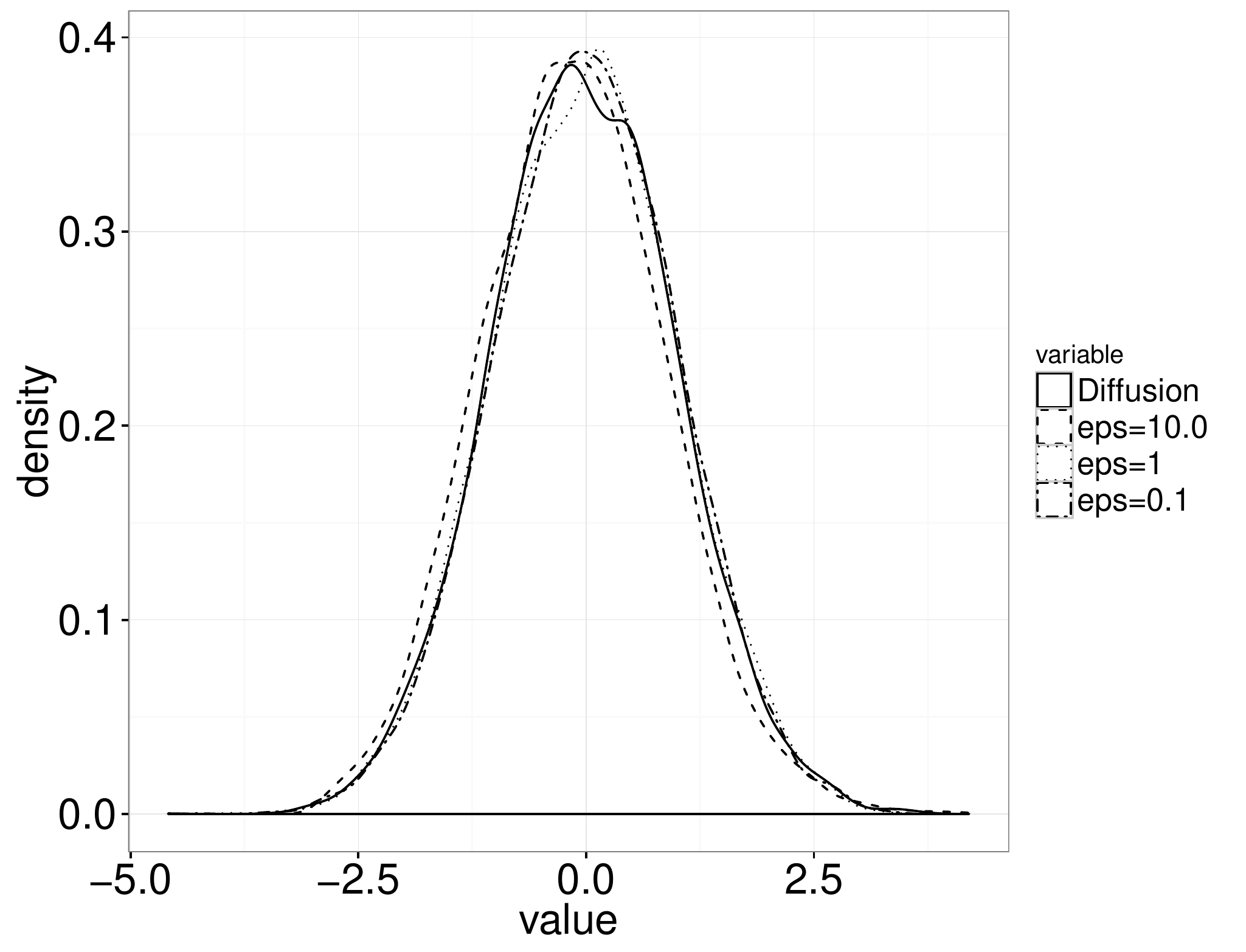}
 \caption{$T=50$}
\end{subfigure}
    \caption{Comparison of distribution of Zig-Zag versus It\^{o} diffusion processes.  The dashed lines denote Zig Zag process $X^\epsilon(t)$ with switching rate~\eqref{eq:switching_example} for different $\epsilon$ while the solid line gives the density of the diffusion process given by~\eqref{eq:langevin_example}.}
    \label{fig:zig_zag_diffusion}
  \end{figure}
 \end{example}


\section{Effective Sample Size for the Zig-Zag process}
\label{sec:ess}

Provided that a central limit theorem holds, for large $T$, the variance of the estimator $\pi_T(f)$ is given to leading order by $T^{-1}\sigma^2_f$, where $\sigma^2_f$ is the asymptotic variance for the observable $f$.  Suppose we wish to obtain an approximation of $\pi(f)$ within a given error tolerance $\epsilon^2$ (in the sense of mean-square error), one can obtain an estimate of the amount of time $T$ that the Zig-Zag process must be simulated, namely
\begin{equation}
\label{eq:running_time}
	T \approx \frac{\sigma^2_f}{\epsilon^2}.
\end{equation}
In general, \eqref{eq:running_time} does not reflect the true cost of simulating the Zig-Zag sampler.  Indeed, as with all continuous time processes, one can accelerate the mixing of a process simply by introducing a time change $Z^a(t) = Z(at)$, for $a > 0$.   In reality, introducing such a time change will increase the number of switches which occur per unit time,  thus increasing the computational effort required to simulate the process up to a given final time $T$.  

Assume that $Z(t)$ is simulated using the direct method (see Algorithm~\ref{alg:zigzag1} in Appendix \ref{sec:simulation_app}). The switching times are determined by a Poisson process with inhomogeneous rate $\int_0^t \lambda(X(s), \Theta(s))\,ds$.  Therefore, the average number of switches occurring in time $[0,T]$ is given by
$$
	N(T) := \mathbb{E}\left[\int_0^T \lambda(X(s), \Theta(s))\,ds\right]
$$
To quantify the average computational cost of simulating a Zig-Zag sampler we introduce the average switching rate $N_{S} = \lim_{t\rightarrow \infty} t^{-1}N(t)$, which measures the average number of switches occurring per unit time.  Since $Z(t)$ is ergodic, then we have that
\begin{equation}
\label{eq:NS}
\begin{aligned}
N_S &=  \lim_{t\rightarrow \infty} \frac{1}{t}\int_0^t \mathbb{E}\left[\lambda(X(s), \Theta(s))\right]\,ds\\ 
 	&=\frac{1}{2}\sum_{\theta=\pm 1}\int_{\mathbb{R}}\lambda(x, \theta)\pi(x)\,dx\\
	&= \frac{1}{2}\int_{\mathbb{R}} \left(|U'(x)| + 2 \gamma(x)\right)\pi(x)\,dx.
\end{aligned}
\end{equation}
where we used the explicit formula for $\lambda(x,\theta)$ given in~\eqref{eq:lambda-explicit}.  Thus,  assuming that $N_S$ is finite, after an initial transient period the number of switchings will increase linearly in time with rate $N_S$.   In terms of computational cost per simulated unit time interval, it is clear that  using canonical switching (i.e. $\gamma = 0$) is the cheapest option.  In this case, the average switching rate will be determined entirely by the target distribution.  

For the purpose of comparison with other sampling schemes, it would be ideal to obtain an expression for the variance of the estimator $\frac{1}{T}\int_0^T f(X_s)\,ds$ as a function of the number of switches required to simulate the Zig-Zag process up to time $T$.   For large $T$ the average number of switches that occurred over $[0,T]$ is approximately $T N_S$ where $N_S$ is given by \eqref{eq:NS}.  Over large time-scales the variance of the estimator $\pi_T(f)$ is thus given  (for the canonical switching rates, $\gamma = 0$), by
$$
	\mbox{Var}\left[\pi_T (f)\right] \approx \frac{\sigma^2_f N_S}{N(T)} = \frac{1}{N(T)}\int_{\mathbb{R}} |U'(x)|\psi^2(x)\pi(x)\,dx \int_{\mathbb{R}}|U'(x)|\pi(x)\,dx,
$$
where $N(T)$ is the number of switches that occured up to time $T$ and $\psi$ is given by \eqref{eq:expression-psi}. 

A useful measure of the effectiveness of a sampling scheme is the \emph{effective sample size} (ESS), which provides a measurement of the equivalent number of IID draws from $\pi$ which would be required to obtain an estimate for $\pi(f)$ with similar variance. For the Zig-Zag sampler, it is natural to define the ESS as follows
\begin{equation}
\label{eq:ESS}
	ESS := \frac{\mbox{Var}_{\pi}[f]}{\Var[\pi_T(f)]} = \frac{\mbox{Var}_{\pi}[f]}{\sigma^2_f N_S} N(T).
\end{equation}
This expression provides a far more natural measure of the effectiveness of the Zig-Zag sampler than e.g.~\eqref{eq:running_time}.  In particular, it is trivial to check that~$\mbox{Var}_{\pi}[f]/(\sigma^2_f N_S)$ is invariant under time rescaling $t \rightarrow a t$, for $a > 0$.  The use of the number of switches $N(T)$ as a measure of computational cost is also well-justified.  One can see from Algorithm \ref{alg:zigzag1} that this coincides with the number of evaluations of the gradient of the log target distribution $U(x)$, which in high dimensions, or in the large data regime for Bayesian inference problems (as considered in \cite{BierkensFearnheadRoberts2016}) would be the most expensive operation required to compute the next term in the event chain. The ESS is linearly increasing with $N(T)$ by a factor equal to $\Var_{\pi}[f]/( \sigma_f^2 N_S)$, which determines the efficiency of the Zig-Zag sampler. 

\begin{example}[Moments of Gaussian distribution]
Consider the problem of computing moments $x^k$ of the Gaussian distribution $\mathcal{N}(0, \nu^2)$, where $k$ is a natural number.   In this case, we can compute the effective switching rate to be $N_S = (2\pi \nu^2)^{-1/2}$, so that using the expression for the asymptotic variance obtain in Example \ref{ex:gaussian-moments} we have for $k$ odd
\begin{equation}
\label{eq:ess_gaussian1}
	\frac{ESS}{N(T)} = \frac{\nu^{2k}\sqrt{2\pi \nu^2}(2k-1)!!}{\frac{\nu^{2k+1}}{\sqrt{2\pi}}2^{k+2}\left(\frac{2k!}{k+1} - \frac{1}{2}((k-1)/2)!)^2\right)} = \frac{2\pi (2k-1)!!}{2^{k+2}\left(\frac{2k!}{k+1} - \frac{1}{2}((k-1)/2)!)^2\right)},
\end{equation}
which is independent of $\nu$.  A tedious calculation reveals that  $ESS > N(T)$, for all such $k$.  A similar computation gives, for $k$ even
\begin{equation}
\label{eq:ess_gaussian2}
	\frac{ESS}{N(T)} = 2\pi\frac{ (2k-1)!! - ((k-1)!!)^2}{\frac{8(k!)^2}{2^k((k/2)!)^2} + \frac{8k!(2^k-k-2)}{k+1}}.
\end{equation}
Evaluating numerically the first few moments using~\eqref{eq:ess_gaussian1} and \eqref{eq:ess_gaussian2} we obtain
\begin{center}
 \begin{tabular}{c|c|c|c|c|c|c}
 \label{tab:ess_table}
  $k$ & 1 & 2 & 3 & 4 & 5 & 6 \\
  \hline
  $ESS/N(T)$ & $1.5708$ & $1.5708$ & $1.1781$ & $1.32278$ & $1.22073$ & $1.33459$
 \end{tabular}
 \end{center}
we see that the relation $ESS > N(T)$ appears to hold for general $k$.  This demonstrates a non-intuitive phenomenon:  the effective sample size of the Zig-Zag process is higher than the number of IID samples.  Thus an ergodic average generated from a trajectory of the Zig-Zag process with $N$ switches will tend to have lower variance than a Monte Carlo average of $N$ IID samples of $\pi$. 
\\\\
To demonstrate the performance of the Zig-Zag sampler, we generate $10^5$ independent realisations of the process ergodic with respect to $\mathcal{N}(0,4)$, and in Figure \ref{fig:gaussian_moment} plot the variance for estimators of the first two moments, as a function of $N$ (the maximum number of switches).  We also plot the variance for a MC average generated from IID samples, as well as for a Random Walk Metropolis-Hastings (RWMH) scheme with manually tuned step-size.  We see that even after manually tuning the step-size of the RWMH chain, the asymptotic variance of the corresponding estimator is still an order of magnitude higher that that of the IID chain and Zig-Zag sampler. In both cases, the ratio of variances for the Zig-Zag sampler and IID average is constant, independent of $N$, as predicted by~\eqref{eq:ess_gaussian1} and~\eqref{eq:ess_gaussian2}. 
\begin{figure}[!ht]
\begin{subfigure}[b]{0.45 \textwidth}
 \includegraphics[width=\textwidth]{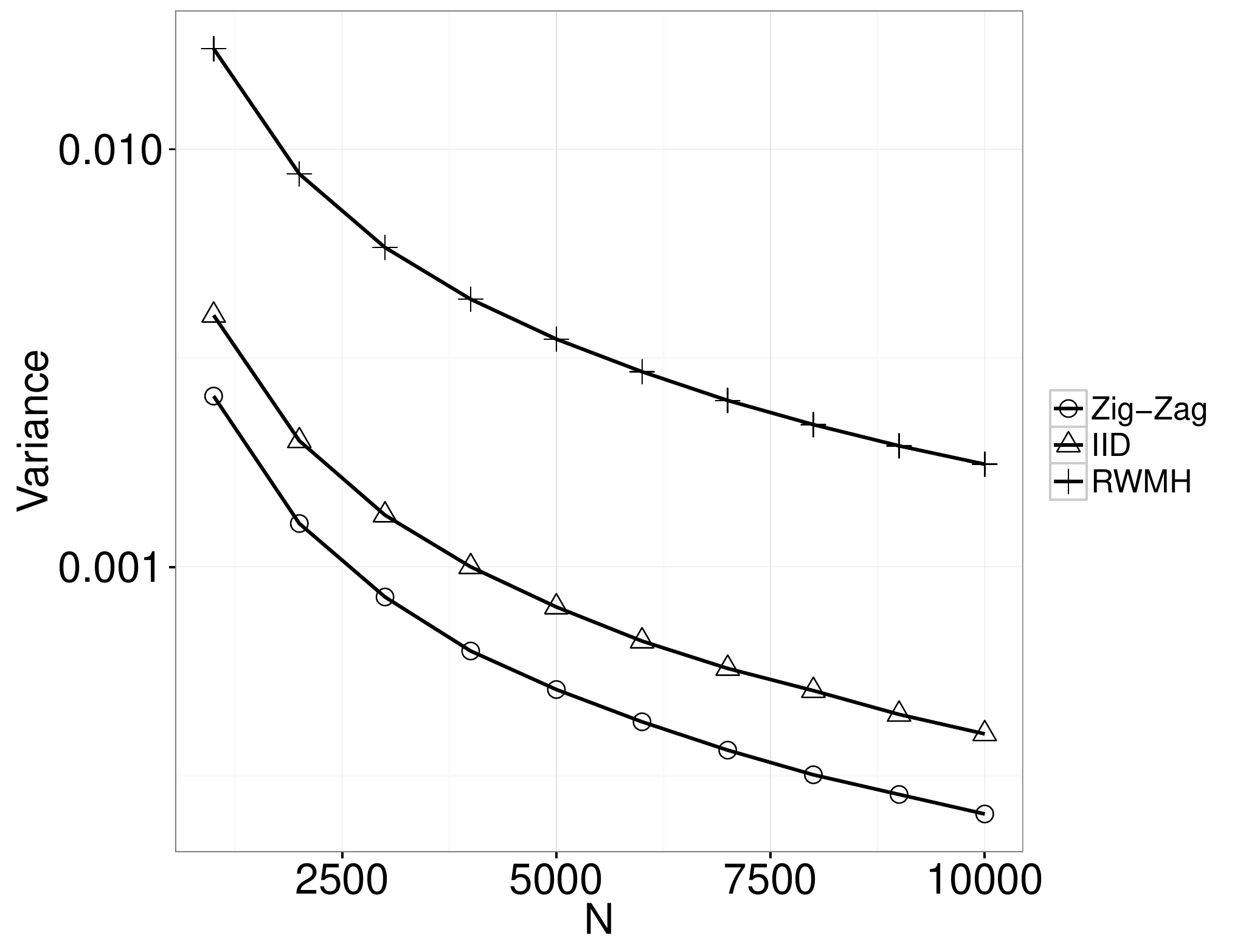}
 \caption{Variance for $f(x) = x$.}
\end{subfigure}
    \hfill
\begin{subfigure}[b]{0.45 \textwidth}
 \includegraphics[width=\textwidth]{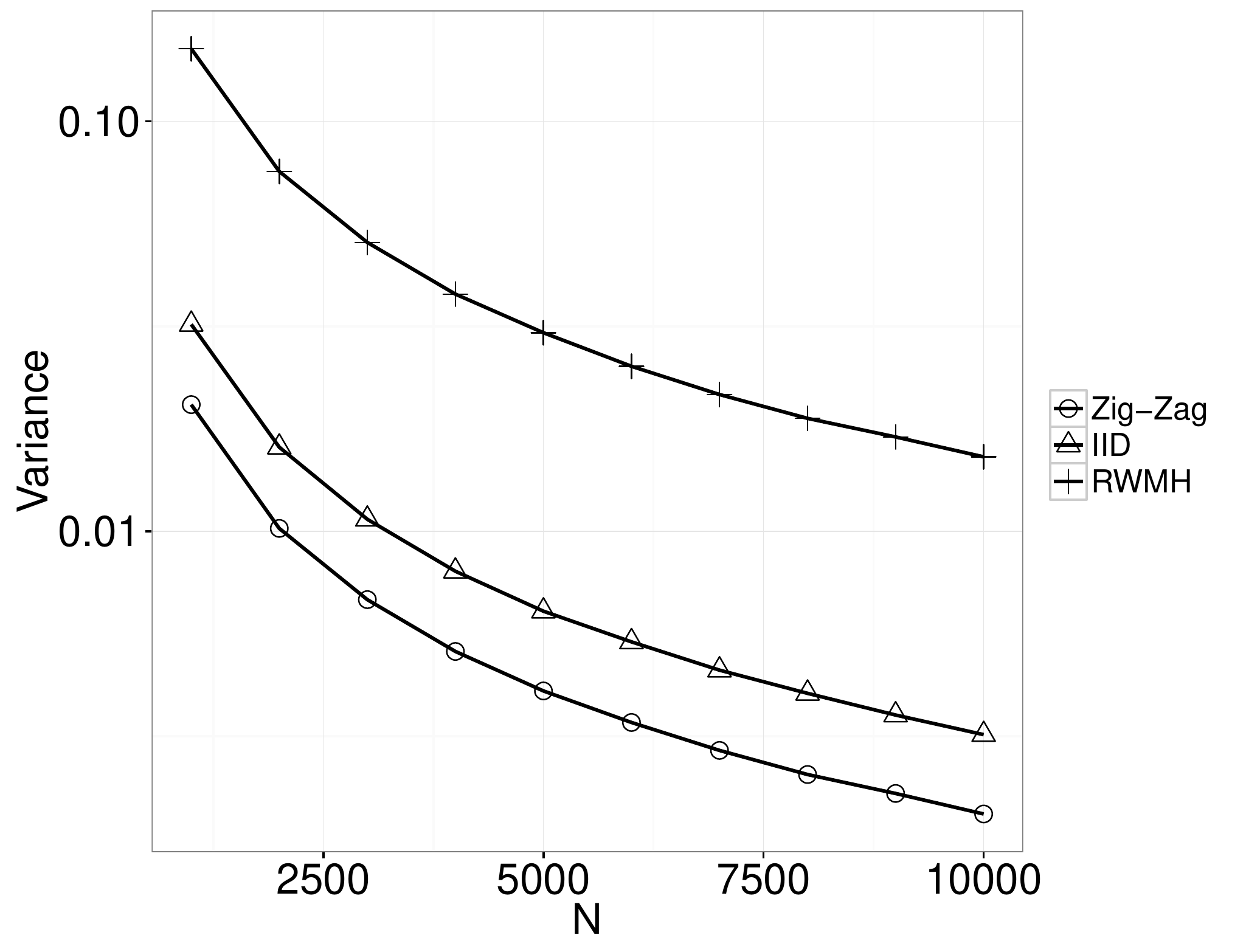}
 \caption{Corresponding plot for $f(x) = x^2$.}
\end{subfigure}
\caption{\label{fig:gaussian_moment}Variance of estimator $\pi_T(f)$ of $f(x)=x$ and $f(x)=x^2$ respectively, as a function of number of switches.  For comparison, the variance of Monte Carlo estimator using IID samples and a tuned Random-Walk-Metropolis-Hastings chain are also displayed.}
  \end{figure}
  \end{example}
\noindent
The fact that the Zig-Zag sampler is able to achieve effective sample sizes which beat IID is a property which is closely tied to the non-reversible nature of the Zig-Zag process. While we have demonstrated this property for the Gaussian case, one should not interpret this as a general result.  Indeed, in the following example we repeat the above experiment  for the Student t-distribution, and we show that although the Zig-Zag sampler  outperforms the corresponding RWMH chain, it will not have ESS higher than that of an IID chain.
\begin{example}[Moments of Student t-distribution]
Following Example~\ref{ex:student-t-moments}, we consider once again the problem of the first moment of the Student t-distribution with $\nu$ degrees of freedom.  In Figure~\ref{fig:student_moment} we plot the variance of estimates for the first moment obtained from the Zig-Zag process using canonical switching rate~\eqref{eq:switching_studentt}, for $\nu =4$, $6$ and $8$.  Each point is generated from $M = 10^5$ independent realisations of the process.  Note that for the observable $f(x) = x$, Assumption~\ref{ass:unimodal} holds for each value of $\nu$.  As in the previous example, we also plot the variance of a Monte-Carlo estimator generated from $M$ IID samples, as well a from a manually tuned RWMH chain.  
\begin{figure}[!ht]
\begin{subfigure}[b]{0.32 \textwidth}
 \includegraphics[width=\textwidth]{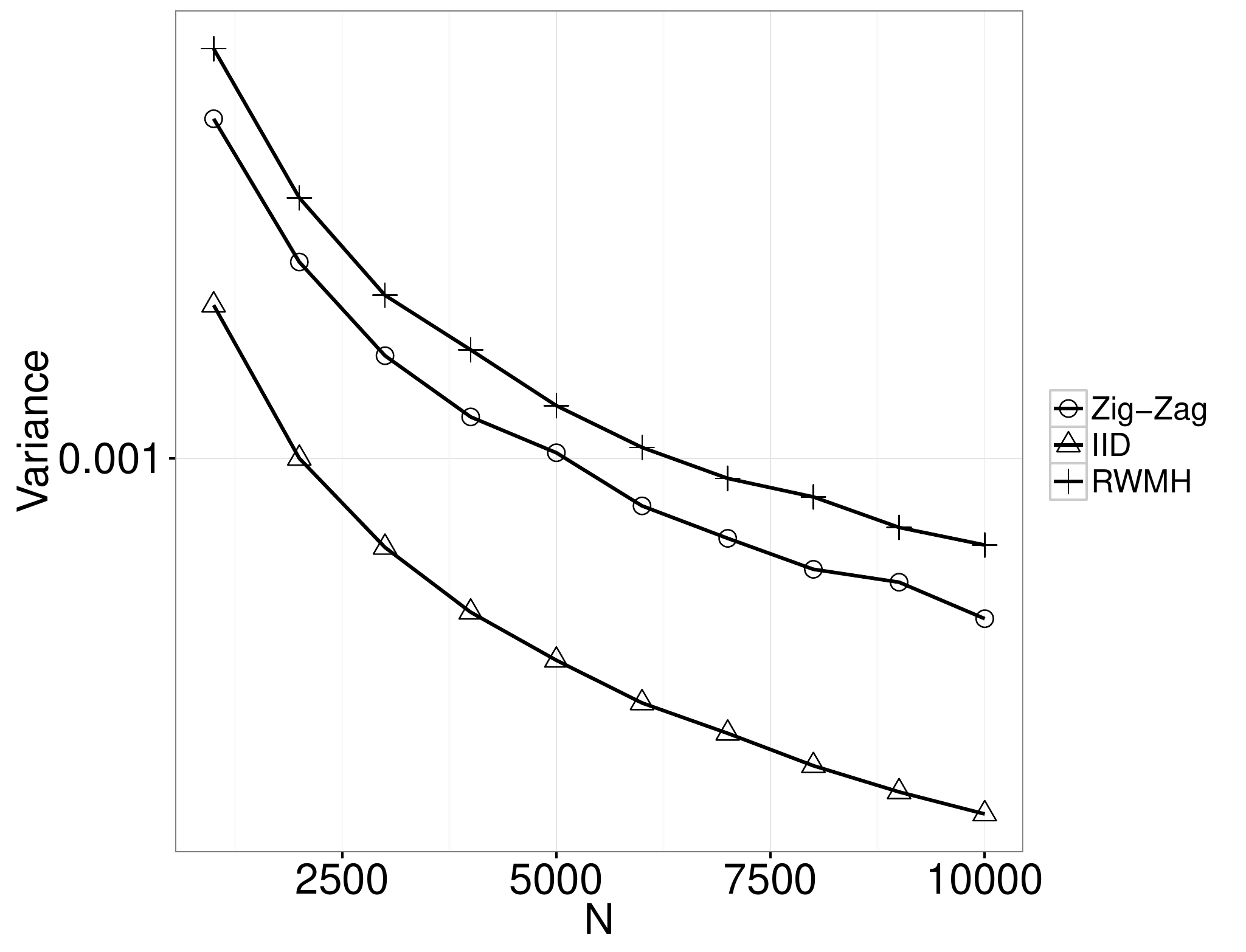}
 \caption{$\nu=4$.  }
\end{subfigure}
\hfill
\begin{subfigure}[b]{0.32 \textwidth}
 \includegraphics[width=\textwidth]{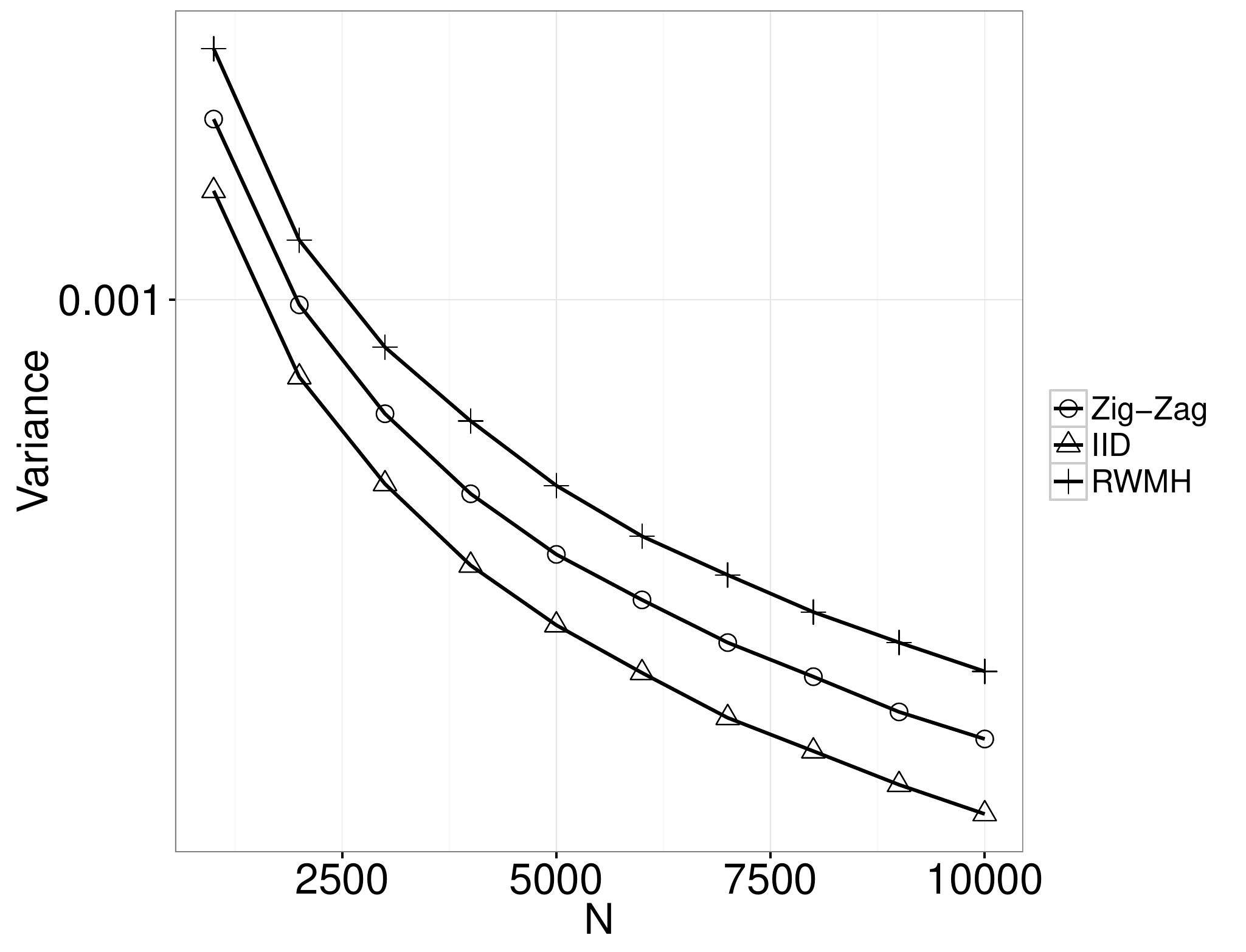}
 \caption{$\nu=6$.  }
\end{subfigure}
\hfill
\begin{subfigure}[b]{0.32 \textwidth}
 \includegraphics[width=\textwidth]{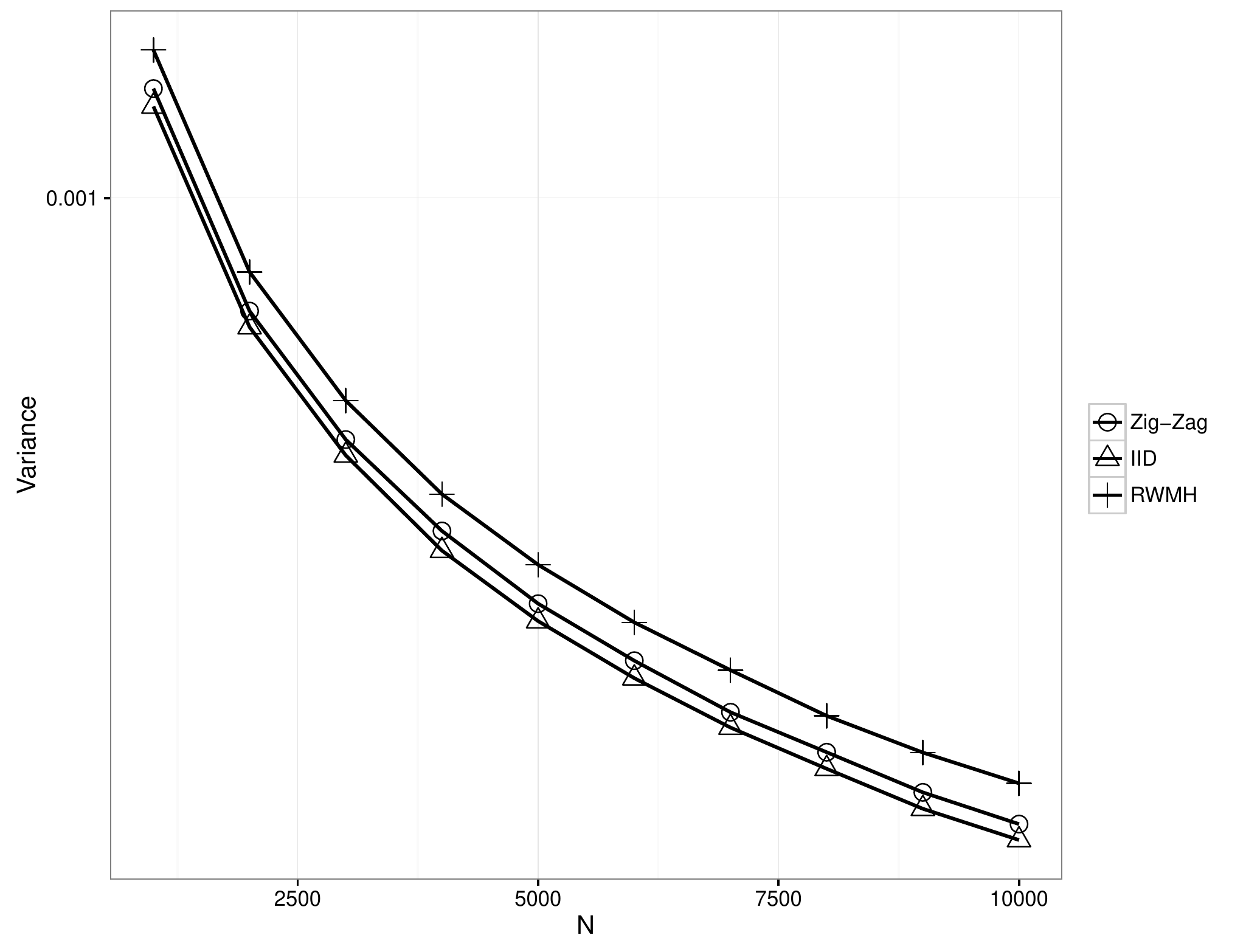}
 \caption{$\nu=8$.  }
\end{subfigure}
\caption{\label{fig:student_moment}Variance of estimator $\pi_T(f)$ of $f(x)=x$ respectively, as a function of number of switches for the student t-distribution.  For comparison, the variance of Monte Carlo estimator using IID samples and a tuned Random-Walk-Metropolis-Hastings chain are also displayed.}
 \end{figure}
In this case the effective sample size of the Zig-Zag sampler will not be  higher than that of the IID estimator, in general.  However, as the degrees of freedom $\nu$ goes to infinity, the target distribution becomes increasingly Gaussian, and for sufficiently large $\nu$, the Zig-Zag sampler will exhibit lower variance than the corresponding IID scheme.  
\end{example}

\appendix

\section{}

\subsection{Proof of Proposition~\ref{prop:well-posedness-stationary-zigzag}}
\label{app:well-posedness}

Because $\lambda$ is locally bounded, \cite[Assumption 3.1]{Davis1984} is satisfied, and a piecewise deterministic Markov process can be constructed as described in \cite{Davis1984}. Then, by \cite[Theorem 5.5]{Davis1984}, $L$ is the extended generator. The Feller property is established by tracing the proof of \cite[Proposition 4]{bierkensroberts2015}, for which only continuity of $\lambda$ is required.
Since $\lambda$ is continuous and because $\lambda(x,\theta) > 0$ for $\theta x \geq x_0$, we have in fact that, for any $x_1 > x_0$,  there exists a $c > 0$ such that 
\[ \lambda(x,\theta) \geq c \quad \mbox{for all $(x,\theta)$ satisfying $x_0 \leq \theta x \leq x_1$.}\]
The proof that compact sets are petite is now a straightforward adaptation of the proof of \cite[Lemma 15]{bierkensroberts2015}, and a Markov process with this property is $\varphi$-irreducible; in particular there exists at most a single invariant measure. The stationarity of $\mu$ is established in \cite[Proposition 5]{bierkensroberts2015}.

\subsection{Technical results towards the CLT}
\label{app:technical}
The following lemma is a continuous time variant of \cite[Exercise 2.4.6]{Durrett1996}.

\begin{lemma}
\label{lem:CLT-randomtime}
 Let $(Y_i)$ be sequence of i.i.d. mean zero random variables with $\E[Y_i^2] = \sigma^2 < \infty$. Suppose $a : [0,\infty) \rightarrow \N$ such that $\lim_{t \rightarrow \infty} a(t) = \infty$ and  $(N(t))_{t \geq 0}$ is a random process such that $\lim_{t \rightarrow \infty} \frac {N(t)}{a(t)} = 1$ in probability. Then
 \[ \frac 1 {\sqrt{a(t)}} \sum_{i=1}^{N(t)} Y_i \Rightarrow \mathcal N(0,\sigma^2) \ \mbox{as $t \rightarrow \infty$}.\]
\end{lemma}

\begin{proof}
Let $\varepsilon > 0$ and $\gamma > 0$. Let $\beta = \varepsilon \gamma^2/(2 \sigma^2)$. Pick $T > 0$ such that for all $t \geq T$,  $|N(t)/a(t) - 1| > \beta$ with probability at most $\varepsilon /2$. For fixed $t \geq T$,  let $\Omega(t)$ denote the event in which $|N(t)/a(t) - 1| \leq \beta$.
On $\Omega(t)$, $|N(t)- a(t)| \leq \lfloor \beta a(t) \rfloor \leq \beta a(t)$. By Kolmogorov's maximal inequality,
\begin{align*} &  \P \left(  \frac 1 {\sqrt{a(t)}}\left|\sum_{i=1}^{N(t)} Y_i  - \sum_{i=1}^{a(t)} Y_i\right| \geq \gamma \right) \leq \frac {\varepsilon}{2} + \P\left( \1_{\Omega(t)} \frac 1 {\sqrt{a(t)}}\left|\sum_{i=1}^{N(t)} Y_i  - \sum_{i=1}^{a(t)} Y_i\right| \geq \gamma \right) \\
& \leq \frac {\varepsilon}{2} +  \P \left( \sup_{m \in \{1, \dots, \lfloor \beta a(t)\rfloor \}} \frac 1 {\sqrt{a(t)}} \left| \sum_{i=1}^m Y_i \right| \geq \gamma\right)  \leq \frac {\varepsilon}{2} + \frac {\beta \sigma^2 a(t) }{\gamma^2 a(t)}  \leq \varepsilon.
\end{align*}
This establishes that $\frac 1 {\sqrt{a(t)}} \left(\sum_{i=1}^N(t) Y_i - \sum_{i=1}^{a(t)} Y_i \right)$ converges in probability to 0. The stated result now follows from the classical central limit theorem applied to $\frac 1 {\sqrt{a(t)}} \sum_{i=1}^{a(t)} Y_i$.
\end{proof}


\begin{proof}[Proof of Lemma~\ref{lem:martingale-expressions}]
Since $\phi \in \mathcal D(L)$ it follows that $M$ is a local martingale. 
Due to stationarity we have
\[ \E_{\mu} |\phi(Z(t))| = \E_{\mu} |\phi(Z(0)| = \mu(|\varphi|) < \infty \]
and 
\[ \E_{\mu} \left| \int_0^t g(Z(s)) \ d s \right| \leq \int_0^t \E_{\mu} \left| g(Z(s)) \right| \ d s = t \mu(|g|) < \infty,\]
where we used that $|g| \leq f$ and $\mu(f) <\infty$ by Proposition~\ref{prop:solution-Poisson-equation}. It follows that $M$ is a martingale.
We have
\begin{align*}
 M(t) & = \phi(Z(t)) - \phi(Z(0)) - \int_0^t L\phi(Z(s)) \ d s \\
 & = \int_{0}^{t} \Theta(s) \phi'(Z(s)) \ d s + \sum_{i=1}^{N(t)} \{ \phi(Z(T_i))) - \phi(Z(T_i-)) \} \\
 & \quad \quad - \int_0^t \left\{  \Theta(s) \phi'(Z(s)) + \lambda(Z(s)) \left(\phi(X(s), \Theta(s)) - \phi(X(s), -\Theta(s))\right) \right\} \ d s \\
 & = \sum_{i=1}^{N(t)}  \{ \phi(X(T_i), \Theta(T_i)) - \phi(X(T_i), -\Theta(T_i) \} \\
 & \quad \quad - \int_0^t \lambda(Z(s)) \left(\phi(X(s), \Theta(s)) - \phi(X(s), -\Theta(s))\right) \ d s  \\
 & = - 2 \sum_{i=1}^{N(t)} \psi(Z(T_i))+ 2 \int_0^t \lambda(Z(s)) \psi(Z(s)) \ d s,
\end{align*}
where $\psi(x) = \half(\phi(x,+1) - \phi(x, -1))$.
Using \cite[Theorem 26.6 (vii), (viii)]{Kallenberg2002} the quadratic variation of $M$ and predictable quadratic variation are given by the stated expressions.
\end{proof}

In Lemma~\ref{lem:martingale-expressions} we introduced the function $\psi : \R \rightarrow \R$. In the following lemma we collect some useful properties of this function.

\begin{proof}[Proof of Lemma~\ref{lem:expression-psi}]
Assume without loss of generality that $\mu(g) = 0$.
Writing out the relation $L \phi(x,\theta)  = - g(x,\theta)$ for $\theta = \pm 1$ and adding the two equations gives
\[ \frac{d \phi(x,+1)}{d x} - \frac{d \phi(x,-1)}{d x} - (\lambda(x,+1) - \lambda(x,-1)) (\phi(x,+1) - \phi(x,-1)) = - (g(x,+1) + g(x,-1))/2,\]
i.e.
\[ \psi'(x) - U'(x) \psi(x) = -(g(x,+1) + g(x,-1))/2.\]
This equation may be solved to give 
\begin{equation} \label{eq:solution-psi} \psi(x) = \frac{c}{\pi(x)} + \frac 1 {2 \pi(x)} \int_x^{\infty} \left\{ g(\xi,+1) + g(\xi,-1) \right\}\pi(\xi) \ d \xi, \quad x \in \R.\end{equation}
It remains to verify that the constant $c$ vanishes. By Proposition~\ref{prop:solution-Poisson-equation}, we have $|\phi| \leq c_0(V+1)$ and hence
\[ |\psi(x)| = |\phi(x,+1) -\phi(x,-1)| \leq c_0(V(x,+1) + V(x,-1) + 2).\]
By the assumption that $\pi(x) V(x,\pm 1) \rightarrow 0$, it therefore follows that $\pi(x) \psi(x) \rightarrow 0$ as $|x| \rightarrow \infty$. 
Multiplying~\eqref{eq:solution-psi} by $\pi$, we have that
\[ \pi(x) \psi(x) = c + \half \int_x^{\infty} \{ g(\xi, +1) + g(\xi,-1)\} \pi(\xi) \ d \xi \rightarrow c \quad \mbox{as $|x| \rightarrow \infty$},\]
so that necessarily $c = 0$.

Now suppose for some $\delta \in \R$, that $|x|^{\delta} \pi(x) \rightarrow 0$ as $|x| \rightarrow \infty$ and~\eqref{eq:condition-psi-asymptotics} holds.
Then since $h(x) := \int_{-\infty}^x \{  g(\xi,+1) +  g(\xi,-1)\} \pi(\xi) \ d \xi \rightarrow 0$ as $x \rightarrow \pm \infty$, using l'H\^opital's rule gives
\[ \frac{\psi(x)}{x^{\delta}} = \frac {h(x)}{x^{\delta} \pi(x)} \sim \frac{h'(x)}{\delta x^{\delta - 1} \pi(x) + x^{\delta} \pi'(x)} = \frac{(g(x,+1) + g(x,-1))\pi(x)}{\underbrace{\delta x^{\delta -1}\pi(x)}_{\rightarrow 0} + x^{\delta} \pi'(x)} \rightarrow 0 \quad \mbox{as $|x| \rightarrow \infty$}.\]
\end{proof}

\subsection{Equivalence of expressions for asymptotic variance}
\label{sec:equiv_var}
A natural question to ask is whether the two expressions for asymptotic variance, given by (\ref{eq:asvar-alternative}) and (\ref{eq:zigzag-asymptotic-variance}) are equivalent in cases where both expressions are valid.  Suppose for an observable $g$ such that $\pi(g) = 0$,
\begin{equation}
\label{ass:equiv_var1}
\lim_{x\rightarrow \pm\infty}e^{-U(x)}\left(\int_{0}^{x}g(y)\,dy\right)^2 = 0,
\end{equation}
and
\begin{equation}
\label{ass:equiv_var2}
\lim_{x\rightarrow \pm\infty} e^{U(x)}\left(\int_{-\infty}^{x}g(y)e^{-U(y)}\,dy\right)^2 = 0.
\end{equation}
Assuming that (\ref{ass:equiv_var1}) and (\ref{ass:equiv_var2}) hold, and that the potential $U$ satisfies $U(0) = 0$, then we can show that both expressions are equal.   
Considering the term 
\begin{align*}
	\int_{0}^{\infty}& U'(x)e^{U(x)}\left(\int_{-\infty}^{x} g(y)e^{-U(y)}\,dy\right)^2\,dx \\
	&= -2\int_{0}^{\infty} g(x)\left(\int_{-\infty}^{x} g(y)e^{-U(y)}\,dy\right)\,dx -  \left(\int_{-\infty}^{0} e^{-U(y)}g(y)\,dy\right)^2,
\end{align*}
where we use (\ref{ass:equiv_var2}) to eliminate the contribution due to the upper integration limit.  Similarly, we have
\begin{align*}
\int_0^\infty & U'(x)e^{-U(x)}\left(\int_{0}^x g(y)\,dy\right)^2\,dx\\
&= 2\int_{0}^\infty e^{-U(x)}g(x)\left(\int_{0}^{x} g(y)\,dy\right)\, dx - \lim_{x\rightarrow \infty} e^{-U(x)}\left(\int_0^x g(y)\,dy\right)^2,
\end{align*}
for which the second term is zero, by (\ref{ass:equiv_var1}).  Exchanging the integrals we obtain
$$
	2\int_{0}^\infty e^{-U(x)}g(x)\int_{0}^{x} g(y)\,dy\, dx = 2\int_0^\infty g(x) \int_{x}^{\infty} e^{-U(y)}g(y) \,dy\,dx
$$
Since $\pi(g) = 0$, it follows that
$$
	\int_x^\infty e^{-U(y)}g(y)\,dy = -\int_{-\infty}^x e^{-U(y)}g(y)\,dy,
$$
and so
$$
2\int_{0}^\infty e^{-U(x)}g(x)\int_{0}^{x} g(y)\,dy\, dx = -2\int_0^\infty g(x) \int_{-\infty}^{x} e^{-U(y)}g(y) \,dy\,dx
$$
so that
\begin{equation}
\label{eq:equiv_var1}
\begin{aligned}
\int_{0}^{\infty} & U'(x)e^{U(x)}\left(\int_{-\infty}^{x} g(y)e^{-U(y)}\,dy\right)^2\,dx \\ = &\int_0^\infty  U'(x)e^{-U(x)}\left(\int_{0}^x g(y)\,dy\right)^2\,dx -  \left(\int_{-\infty}^{0} e^{-U(y)}g(y)\,dy\right)^2.
\end{aligned}
\end{equation}
Arguing similarly, one has that
\begin{equation}
\label{eq:equiv_var2}
\begin{aligned}
\int_{-\infty}^{0} & U'(x)e^{U(x)}\left(\int_{-\infty}^{x} g(y)e^{-U(y)}\,dy\right)^2\,dx \\ = &\int_{-\infty}^0  U'(x)e^{-U(x)}\left(\int_{0}^x g(y)\,dy\right)^2\,dx +  \left(\int_{-\infty}^{0} e^{-U(y)}g(y)\,dy\right)^2.
\end{aligned}
\end{equation}
Combining~\eqref{eq:equiv_var1} and~\eqref{eq:equiv_var2} it follows immediately that the expressions for asymptotic variance respectively given by \eqref{eq:asvar-alternative} and \eqref{eq:zigzag-asymptotic-variance} are equal.

\subsection{Proof of Proposition~\ref{prop:langevin-asvar}}
\label{app:langevin-asvar}

Write $P^s$ for the Markov semigroup corresponding to the Langevin diffusion, with generator $A$. By \cite[Corollary 1.9]{KipnisVaradhan1986}, a CLT is satisfied if there exists a constant $c > 0$ such that 
\[ \langle g, f \rangle_{L^2(\pi)} \leq c \langle - A f, f \rangle^{1/2}_{L^2(\pi)} \]
for all $f \in \mathcal D(A)$, where the domain of $A$ is interpreted as corresponding to the domain of the semigroup generator in $L^2(\pi)$. It is sufficient to check this condition for $f$ in the space $C_c^{\infty}(\R)$ of infinitely differentiable functions with compact support, as this is a core for $A$. 
By partial integration on both sides, the above condition then becomes
\[ \langle -\psi, f' \rangle_{L^2(\pi)} \leq c \| f' \|_{L^2(\pi)} \quad \mbox{for all $f \in C_c^{\infty}(\R)$}.\]
which is satisfied for $c = \| \psi\|_{L^2(\pi)}$.
In this case, by \cite[Corollary 1.9]{KipnisVaradhan1986}, the asymptotic variance admits the expression 
\begin{align*}
\widetilde \sigma_g^2 & = 2 \langle \varphi, g \rangle_{L^2(\pi)} = -2 \int_{-\infty}^{\infty} \varphi(x) \left( \frac 1 {\pi(x)} \frac{d}{dx} (\pi(x) \varphi'(x)) \right) \pi(x) \ d x  = 2 \int_{-\infty}^{\infty} (\varphi'(x))^2 \pi(x) \ d x.
\end{align*}
where $\varphi$ satisfies the Poisson equation $A \varphi = -g$. By the Poisson equation for $\varphi$,
\[ \pi(x) \varphi'(x) = \int_{x}^{\infty} \pi(\xi) g(\xi) \ d \xi + c,\]
By a similar argument as in the proof of Lemma~\ref{lem:expression-psi}, using that $\varphi \in \mathcal D(A)$ and hence $\varphi' \in  L^2(\pi)$, it follows that $c = 0$ and hence  $\varphi'(x) = - \psi(x)$.

We now prove the converse.  To this end, suppose that 
\begin{equation}
\label{eq:kipvara_assumption}
	V:= \limsup_{t\rightarrow \infty}\frac{1}{t}\mbox{Var}_{\pi}\left(\int_0^t g(X_s)\,ds\right) = \int_0^\infty \int_{\R} (P^s g(x))^2 \pi(x) \, dx \,ds < \infty,
\end{equation}
where the equality holds due to \cite[Lemma 2.3]{Cattiaux2011a}.
For any $t > 0$ define 
$$
	g_t := -\int_0^t P^s g \,ds.
$$
Note that $g_t \in \mathcal{D}(A)$ and satisfies 
\begin{equation}
\label{eq:new_resolvent}
A g_t  = (I-P^t)g.
\end{equation}  
We follow the approach of \cite[Theorem 3.3]{Cattiaux2011a}. Below, let $f'$ denote $\frac{d}{dx} f$. Given $s \leq t$,
\begin{align*}
\int \left( g_t' -  g_s' \right)^2 \pi \, dx &=  \int_{\R} -A(g_t - g_s)(g_t - g_s)\,\pi(x) \, dx \\
	&=\int_s^t \int_{\R} (P^s g - P^t g)(P^r g) \pi \, dx\,dr\\
	&=\int_s^t \int_{\R} \left\{ (P^{(r+s)/2} g)^2 - (P^{(r+t)/2} g)^2\right\}  \,  \pi \, dx\,dr\\
	&\leq 2\int_{s}^{\infty}\int_{\R} (P^{r} g)^2 \pi \, dx \, dr. 
\end{align*}
It follows that the family $(g'_t)_{t>0}$ is Cauchy in $L^2(\pi)$, so that it strongly converges to a limit $-\eta \in L^2(\pi)$.  The weak formulation of~\eqref{eq:new_resolvent} is given by
\begin{equation}
\label{eq:weak}
	-\left\langle  g_t', v'\right\rangle_{L^2(\pi)} = \langle g, v\rangle_{L^2(\pi)} - \langle P^t g, v\rangle_{L^2(\pi)},\quad v \in C^\infty_{c}(\mathbb{R}).
\end{equation}
We have $\lim_{t \rightarrow \infty} P^t g  = \pi(g) = 0$, so that by dominated convergence $\langle P^t g, v \rangle_{L^2(\pi)} \rightarrow 0$ as $t \rightarrow \infty$, and thus taking the $t\rightarrow \infty$ limit in~\eqref{eq:weak} gives
$$
	\langle \eta, v'\rangle_{L^2(\pi)} = \langle g, v \rangle_{L^2(\pi)},\quad v\in C^\infty_c(\mathbb{R}).
$$
By the definition of $\psi$, we also have for all $v \in C_c^{\infty}(\R)$ that  $ \langle \psi, v' \rangle_{L^2(\pi)} = \langle g, v \rangle_{L^2(\pi)}$, so that
$\langle (\psi - \eta), v' \rangle_{L^2(\pi)} = 0$.
Hence in the sense of distributions, $(\psi - \eta)' = 0$, from which it follows (see e.g.  \cite[Section 21.4]{KolmogorovFomin1975}) that $\eta = \psi + \mathrm{const}$.
In order for $\eta$ to belong to $L^2(\pi)$, by a similar argument as in the proof of  Lemma~\ref{lem:expression-psi}, the constant should be equal to zero and hence $\psi = \eta \in L^2(\pi)$.

\subsection{Proof of Theorem \ref{thm:diffusive}}
\label{sec:proof_diffusive}
In this section we prove Theorem \ref{thm:diffusive}, following the approach of \cite{fontbona2012quantitative}.  To this end, consider the function
\[
f(x,\theta) :=x+\frac{\epsilon}{2}\frac{\theta}{\gamma(x)}-\frac{\epsilon^{2}}{2}\frac{\theta}{\gamma^{2}(x)}\lambda^{0}(x,\theta), \quad (x,\theta) \in E.
\]
%
Using the fact that 
$$\half\theta(\lambda^{0}(x,+1)+\lambda^{0}(x,-1))=\theta\lambda^{0}(x,\theta)-\half(\lambda^{0}(x,1)-\lambda^{0}(x,-1)) = \theta \lambda^0(x,\theta) - \half U'(x),$$ we obtain
\begin{align*}
& L^{\epsilon}f(x,\theta) \\
& =  \theta-\frac{\epsilon}{2}\frac{\gamma'(x)}{\gamma^{2}(x)}-\epsilon\frac{\lambda^{0}(x,\theta)}{\gamma(x)}\theta-\theta+\frac{\epsilon}{2}\frac{1}{\gamma(x)}\theta[\lambda^{0}(x,-1)+\lambda^{0}(x,+1)]+\epsilon^{2}R_{1}(x,\theta),\\
 & = -\frac{\epsilon}{2}\frac{\gamma'(x)}{\gamma^{2}(x)}-\epsilon\frac{\lambda^{0}(x,\theta)}{\gamma(x)}\theta+\epsilon\frac{1}{\gamma(x)}\left[\theta\lambda^{0}(x,\theta)-\frac{1}{2}(\lambda^{0}(x,+1)-\lambda^{0}(x,-1))\right]+\epsilon^{2}R_{1}(x,\theta)\\
 & =  \epsilon b(x)+\epsilon^{2}R_{1}(x,\theta)
\end{align*}
where 
\[ b(x) =  -\half \left( \frac{\gamma'(x)}{\gamma^{2}(x)} + \frac{U'(x)}{\gamma(x)} \right) \]
and where
$$
	R_{1}(x,\theta) = -\frac{\lambda^{0}(x,\theta)}{2\gamma(x)^2}U'(x) + \frac{\lambda^{0}(x,\theta)\gamma'(x)}{\gamma(x)^3} - \frac{\partial_x\lambda^{0}(x,\theta)}{2\gamma(x)^2},
$$
is a remainder term which is measurable and independent of $\epsilon$.  Defining
\begin{align*}
Y^{\epsilon}(t)  := f(X^{\epsilon}(t), \Theta^{\epsilon}(t)) \quad \mbox{and}  \quad 
j^{\epsilon}(t)  :=  \epsilon b(X^{\epsilon}(t)) + \epsilon^2 R_1(X^{\epsilon}(t), \Theta^{\epsilon}(t)),
\end{align*}
it follows (using that $f$ is in the domain of the extended generator, see \cite[Theorem 5.5]{Davis1984}), that 
$$
	M^{\epsilon}(t):=Y^{\epsilon}(t)-\int_{0}^{t}j^{\epsilon}(s) \, ds,
$$
is a local martingale with respect to the filtration $\mathcal{F}^{\epsilon}_t$
generated by $\lbrace Z^{\epsilon}(t) :  t\in[0,T]\rbrace.$ Similarly,
applying the generator to $g(x,\theta) :=f^2(x,\theta)$, we obtain
\begin{align*}
& L^{\epsilon}g(x,\theta)\\
& =2\theta x-2\frac{\gamma(x)}{\gamma(x)}\theta x+2 \epsilon\left(\half \partial_{x} \left(\frac {x}{\gamma(x)} \right)-\frac{\theta x}{\gamma(x)} \left(  \lambda^{0}(x,\theta) - \half (\lambda^{0}(x,+1)+\lambda^{0}(x,-1)) \right)\right) \\
& \quad \quad \quad \quad \quad \quad +\epsilon^{2}R_{2}(x,\theta) \\
& = \epsilon \left( a(x) + 2 x b(x)\right) +\epsilon^{2}R_{2}(x,\theta),
\end{align*}
where $b(x)$ is as above, $a(x) = \frac {1} {\gamma(x)}$, and $R_{2}(x,\theta)$ can be written as $R_{2} = R^{(1)}_{2} + \epsilon R^{(2)}_{2} + \epsilon^2 R^{(3)}_{2}$, where the terms
\begin{align*}
	R^{(1)}_{2}(x, \theta) & = -\frac{|U'(x)|}{2\gamma(x)^2} - \frac{\theta\gamma'(x)}{2\gamma(x)^3}  \\
	& \quad + x\left(\frac{2\theta\lambda^{0}(x,\theta)^2}{\gamma(x)^2} - \frac{U'(x)\lambda^{0}(x,\theta)}{\gamma(x)^2} + 2\frac{\lambda^{0}(x,\theta)\gamma'(x)}{\gamma(x)^3} - \frac{\partial_x\lambda^{0}(x,\theta)}{\gamma(x)^2}\right), \\
	R^{(2)}_{2}(x, \theta) &   = \frac{3}{2}\frac{\lambda^{0}(x, \theta)\gamma'(x)}{\gamma(x)^4} + \frac{|U'(x)|^2}{4\gamma(x)^3} - \frac{\theta\partial_x\lambda^{0}(x,\theta)}{2\gamma(x)^4}, \quad \mbox{and}  \\
	R^{(3)}_{2}(x, \theta) & = \frac{\lambda^{0}(x,-\theta)^2\lambda^{0}(x,\theta)}{4\gamma(x)^4} - \frac{\lambda^0(x,\theta)^3}{4\gamma(x)^4} - \frac{\theta\lambda^0(x,\theta)^2\gamma'(x)}{\gamma(x)^5} + \frac{\theta\lambda^0(x,\theta)\partial_x \lambda^{0}(x,\theta)}{2\gamma(x)^4},
\end{align*}
are measurable and independent of $\epsilon$.  We thus obtain that
\[
N^{\epsilon}(t):=(Y^{\epsilon}(t))^{2}-\epsilon\int_{0}^{t} \left\{ a(X^{\epsilon}(s)) + 2 X^{\epsilon}(s) b(X^{\epsilon}(s))
-\epsilon R_{2}(X^{\epsilon}(s), \Theta^{\epsilon}(t)) \right\} ds,
\]
is a local martingale with respect to the filtration $\mathcal{F}^{\epsilon}_t$. 
We now decompose the square local martingale $(M^{\epsilon}(t))^{2}$ into
a local martingale term and a remainder. To this end, defining $J^{\epsilon}(t)=\int_{0}^{t}j^{\epsilon}(s)\,ds,$
use integration by parts to obtain
\begin{align*}
(M^{\epsilon}(t))^{2}  = & (Y^{\epsilon}(t))^{2}-2J^{\epsilon}(t)Y^{\epsilon}(t)+(J^{\epsilon}(t))^{2}\\
  = & (Y^{\epsilon}(t))^{2}-2\int_{0}^{t}Y^{\epsilon}(s)j^{\epsilon}(s)\,ds-2\int_{0}^{t}J^{\epsilon}(s)\, dY^{\epsilon}(s)+2\int_{0}^{t}J^{\epsilon}(s)j^{\epsilon}(s)\,ds\\
  = & (Y^{\epsilon}(t))^{2}-2\int_{0}^{t}Y^{\epsilon}(s)j^{\epsilon}(s)\,ds-2\int_{0}^{t}J^{\epsilon}(s)\, dM^{\epsilon}(s) \\
  = & N^{\epsilon}(t)-2\int_{0}^{t}J^{\epsilon}(s) \, dM^{\epsilon}(s)+
  \epsilon \int_0^t \{ a(X^{\epsilon}(s)) + 2 X^{\epsilon}(s) b(X^{\epsilon}(s)) - \epsilon R_2(X^{\epsilon}(s), \Theta^{\epsilon}(s)) \} \, d s \\
  &  \quad - 2 \int_0^t \left(X^\epsilon(s) + \frac{\epsilon}{2}\frac{\Theta^{\epsilon}(s)}{\gamma(X^\epsilon(s))}  -\epsilon^2\frac{\Theta^{\epsilon}(s)}{2\gamma(X^\epsilon(s))^2}\lambda^{0}(X^\epsilon(s), \Theta^\epsilon(s))\right) \\
  & \quad \quad \quad \quad \quad \quad  \times \left( \epsilon b(X^{\epsilon}(s)) + \epsilon^2 R_1(X^{\epsilon}(s), \Theta^{\epsilon}(s)) \right)    \ d s \\
  & = N^{\epsilon}(t) - 2 \int_0^t J^{\epsilon}(s) \ d M^{\epsilon}(s) + \epsilon \int_0^t \left\{  a(X^{\epsilon}(s))  + \epsilon R_3(X^{\epsilon}(s), \Theta^{\epsilon}(s)) \right\} \ d s,
\end{align*}
%
where the terms of order $\epsilon^2$ or higher are collected in the remainder term $R_{3}(x,\theta)$. It follows that 
$$
	H^{\epsilon}(t) := (M^{\epsilon}(t))^{2}-\epsilon\int_{0}^{t} \left\{ a (X^{\epsilon}(s))+\epsilon R_{3}(X^{\epsilon}(s),\theta^{\epsilon}(s)) \right\} \,ds
$$
is a local martingale with respect to $\mathcal{F}^{\epsilon}_t.$  Applying the time change $t\rightarrow t/\epsilon$ we see that 
\[ M^{\epsilon}(t/\epsilon) = f(X^{\epsilon}(t/\epsilon), \Theta^{\epsilon}(t/\epsilon)) - \int_0^t  \left\{ b(X^{\epsilon}(s/\epsilon) ) + \epsilon R_1(X^{\epsilon}(s/\epsilon), \Theta^{\epsilon}(s/\epsilon)) \right\} \ d s \]
and
\[ H^{\epsilon}(t/\epsilon) = (M^{\epsilon}(t/\epsilon))^2 - \int_0^t \left\{ a(X^{\epsilon}(s/\epsilon)) + \epsilon R_3(X^{\epsilon}(s/\epsilon), \Theta^{\epsilon}(s/\epsilon)) \right\} \ d s 
\] 
%
are local martingales with respect to the filtration $\mathcal{F}^{\epsilon}_t :=\mathcal{F}_{t/\epsilon}$, $t \geq 0$.
We now verify the conditions of \cite[Theorem VII.4.1]{EthierKurtz2005} to derive the diffusive limit.  To this end, define
\begin{align*}
A^{\epsilon}(t)  & := \int_{0}^{t} a(X^{\epsilon}(s/\epsilon)) +\epsilon R_{3}(X^{\epsilon}(s/\epsilon), \Theta^{\epsilon}(s/\epsilon))\,ds,\\
B^{\epsilon}(t) & :=  -\frac{\epsilon}{2}\frac{\Theta^{\epsilon}(t/\epsilon)}{\gamma(X^{\epsilon}(t/\epsilon))} \\
& \quad \quad 
+ \frac{\epsilon^{2}}{2}\frac{\Theta^{\epsilon}(t/\epsilon)}{\gamma^{2}(X^{\epsilon}(t/\epsilon))}\lambda^{0}(Z^{\epsilon}(t/\epsilon)) 
+ \epsilon  \int_{0}^{t} \left\{ b(X^{\epsilon}(s/\epsilon) ) + \epsilon R_{1}(X^{\epsilon}(s/\epsilon), \Theta^{\epsilon}(s/\epsilon)) \right\} \,ds,
\end{align*}
as well as the stopping time $\tau_{R}^{\epsilon}:=\inf\left\{ t\geq0\,:\,\left|X^{\epsilon}(t/\epsilon)\right|\geq R\mbox{ or }\left|X^{\epsilon}(t/\epsilon-)\right|\geq R\right\}$.  From our assumptions we have that, for each $R\geq0:$
\[
\sup_{t\leq T\wedge\tau_{R}^{\epsilon}}|A^{\epsilon}(t)-\int_{0}^{t} a(X^{\epsilon}(s))\,ds|\leq\epsilon\sup_{t\leq T\wedge\tau_{R}^{\epsilon}}\left|\int_{0}^{t}R_{3}(X^{\epsilon}(s/\epsilon), \Theta^{\epsilon}(s/\epsilon))\,ds\right|\leq\epsilon TK_{R},
\]
and thus converges to $0$ almost surely as $\epsilon\rightarrow 0$.
Similarly
\[
\sup_{t\leq T\wedge\tau_{R}^{\epsilon}}|B^{\epsilon}(t)-\int_{0}^{t}b(X^{\epsilon}(s))\,ds|\rightarrow0,
\]
almost surely as $\epsilon\rightarrow0$. Finally, noting that $X^{\epsilon}(t)$
is continuous, we have that
\[
\lim_{\epsilon\rightarrow0}\mathbb{E}\left[\sup_{t\leq T\wedge\tau_{R}^{\epsilon}}|X^{\epsilon}(t/\epsilon)-X^{\epsilon}(t/\epsilon-)|^{2}\right]=0,
\]
and similarly for every $R \geq 0$,
\[
\lim_{\epsilon\rightarrow0}\mathbb{E}\left[\sup_{t\leq T\wedge\tau_{R}^{\epsilon}}|A^{\epsilon}(t)-A^{\epsilon}(t-)|^{2}\right] = \lim_{\epsilon\rightarrow0}\mathbb{E}\left[\sup_{t\leq T\wedge\tau_{R}^{\epsilon}}|B^{\epsilon}(t)-B^{\epsilon}(t-)|^{2}\right]=0.
\]

The conditions of \cite[Theorem VII.4.1]{EthierKurtz2005} are satisfied and thus it follows that $(X^{\epsilon}(t/\epsilon))_{t \geq 0}$ converges in distribution to the solution of the martingale problem for the operator $(G, \mathcal{D}(G))$, where $\mathcal{D}(G)= C^2_c(\mathbb{R})$ and for $h \in \mathcal D(G)$, 
\begin{align*}
G h(x)  & =  b(x)\partial_{x}h(x) +\frac{1}{2}a(x)\partial_{x}^2 h(x) =  -\frac{1}{2}\left(\frac{\gamma'(x)}{\gamma^{2}(x)}+\frac{U'(x)}{\gamma(x)}\right)\partial_{x}h(x) + \frac{1}{2 \gamma(x)}\partial_{x}^2 h(x).
\end{align*}
 Since the well-posedness of this martingale problem is equivalent to the existence and uniqueness of a weak solution $(\xi(t))_{t \geq 0}$ for~\eqref{eq:limiting_sde}, the proof is complete.

\section{Simulation of the Zig-Zag process}
\label{sec:simulation_app}

In this section we describe some computational methods for simulating the process $Z(t) = (X(t), \Theta(t))$ and use results from previous sections in analyzing these methods. As with the rest of this paper, we shall focus in particular on the one-dimensional case, referring the reader to \cite{BierkensFearnheadRoberts2016} for specifics of the general case.

\subsection{Direct simulation of the switching times}
 Clearly, it is sufficient to be able to simulate the random switching times $(T_i)_{i\in \mathbb{N}}$.  Indeed, given initial conditions $(x, \theta) \in E$ and switching times $(T_i)_{i \in \mathbb{N}}$, the process $Z(t) = (X(t), \Theta(t))$ is defined for all $t \geq 0$ as follows:
$$
\Theta(t) = (-1)^k \theta,\quad t \in [T_k, T_{k+1}),
$$
and 
$$ 
	X(t) = X(T_k) + (t-T_k)\Theta(T_k),\quad t  \in [T_k, T_{k+1}].
$$
Given the state $(X(T_0), \Theta(T_0)) = (x_0, \theta_0)$ at switching time $T_0$,  the next random switching time is given by $T_1 = T_0 + \tau$ where $\tau$ satisfies 
\begin{equation}
\label{eq:switching_cdf}
	\mathbb{P}[\tau > t] = \exp\left(-\int_0^t \lambda(x_0 + s\theta_0, \theta_0)\,ds\right).
\end{equation}
 In the case where $G(t) = \int_0^t \lambda(x_0 + s\theta_0, \theta_0)\,ds$ has an explictly computable generalised inverse 
 $$
 	H(y) = \inf\left\lbrace t \geq 0  :  G(t) \geq y \right\rbrace, \quad y \in [0,1],
 $$
then applying an inverse transformation, the random variable $\tau = H(-\log u)$, $u \sim U[0,1]$ satisfies~\eqref{eq:switching_cdf}. An algorithm for simulating $Z(t)$ based on this approach is detailed in Algorithm \ref{alg:zigzag1}.   

\begin{algorithm}[H]
\renewcommand{\algorithmicrequire}{\textbf{Input:}}
\renewcommand{\algorithmicensure}{\textbf{Output:}}

\caption{Direct Zig-Zag Sampling}
\label{alg:zigzag1}
\begin{algorithmic}[1]
\REQUIRE Initial condition $(x, \theta) \in E$.\\
\ENSURE The event chain $(T_k, X(T_k), \Theta(T_k))_{k=0}^{\infty}$.
\STATE Set $(T_0, X(T_0), \Theta(T_0))=(0, x, \theta)$.
\FOR{$k=0,1,2,\ldots$}
	\STATE{Draw $u \sim U[0,1]$.}
	\STATE{Set $\tau = H(-\log u)$.}
	\STATE{Set \begin{align*}T_{k+1} &= T_{k} + \tau,\\ 
							 X(T_{k+1}) &= X(T_{k}) + \tau \Theta(T_{k}),\\
							 \Theta(T_{k+1})&= -\Theta(T_k).\end{align*}}
\ENDFOR
\end{algorithmic}
\end{algorithm} 

The computational cost of Algorithm \ref{alg:zigzag1} clearly depends on the switching intensity, i.e. a Zig-Zag sampler with a higher switching intensity will require more computational cost to be simulated up to a fixed time $T$.   Indeed, while the Zig-Zag sampler does not reject samples like a Metropolis-Hastings scheme,  frequent switching will cause the process $Z(t)$ to mix slowly.


\begin{example}[Sampling from a Gaussian Distribution $\mathcal{N}(0,\nu^2)$]
A straightforward calculation shows that, given $(x, \theta)\in E$ the generalised inverse of $G(t) = \nu^{-2}\int_0^t \max(0, x_0 + s\theta)\,ds$ can be written explicitly as 
$$
	H(z) = \begin{cases}
				-\theta x + \sqrt{2\nu^2 z}\quad &\mbox{ if } \theta x \leq 0	\\
				-\theta x + \sqrt{x^2 + 2 z \nu^2} &\mbox{ otherwise}.
		   \end{cases},
$$
for $z > 0$.   In this case, the average switching rate is then given by $N_S = (2\pi \nu^2)^{-1/2}$.  


\end{example} 

\begin{example}[Sampling from a Student t-distribution] 
It is also possible to sample from a Student t-distribution with $\nu$ degrees of freedom, i.e.
\begin{equation} \label{eq:student-t-density}	\pi(x) \propto ( 1 + \frac{x^2}{\nu})^{-\frac{\nu+1}{2}},
\end{equation}
using the direct Zig-Zag sampling approach.  For this distribution, the canonical switching function is given by
\begin{equation}
\label{eq:switching_studentt}
	\lambda(x, \theta) =\begin{cases}
							\frac{(\nu+1)\theta x }{\nu + x^2}\quad &\mbox{ if } \theta x \geq 0,\\
							0	\quad &\mbox{ otherwise}.
						\end{cases}
\end{equation}
Given $(x, \theta) \in E$, the generalised inverse of  $G(t)=\int_0^{t} \lambda(x + \theta s, \theta)\,ds$ can be written as
$$
H(z) = \begin{cases}
				-\theta x + \sqrt{\left(-1 + \exp\left(\frac{2z}{1+\nu}\right)\right)\nu}\quad &\mbox{ if } \theta x \leq 0	\\
				-\theta x + e^{\frac{z}{1+\nu}}\left(x^2 + \nu - \nu\exp\left(-\frac{2z}{1+\nu}\right)\right)^{1/2} &\mbox{ otherwise}.
		   \end{cases}.
$$
The average switching rate is equal to the normalization constant for~\eqref{eq:student-t-density}, i.e. $N_S = \frac{\Gamma((\nu+1)/2)}{\sqrt{\nu \pi} \Gamma (\nu/2)}$.  The resulting process will be ergodic with respect to the target distribution $\pi$, for all $\nu > 0$. Conditions under which a central limit theorem holds will be studied in  Section \ref{sec:clt_examples}.  
\end{example}
\subsection{Sampling with Poisson Thinning}
In general we will not be able to compute the generalized inverse of $G$ explictly. In many cases however, it is possible to obtain an upper bound $\Lambda(t ; x, \theta_0)$ such that  $m(t) := \lambda(x_0 + \theta_0 t, \theta_0) \leq \Lambda(t; x_0, \theta_0)$, for all $t \geq 0$, $(x_0,\theta_0) \in E$, and where $\Lambda(t)$ has an explicitly computable inverse $\widetilde{H}$.  In this case, one can simulate the random switching times using a standard Poisson thinning approach \cite{lewis1979simulation}.   Using the upper bound $\Lambda(t; x_0, \theta_0)$ a candidate switching time $T_1 = T_0 + \tau$ is generated, such that 
$$
	\mathbb{P}[\tau > t] = \exp\left(-\int_0^t \Lambda(s; x_0, \theta_0)\,ds\right).
$$
A switch (i.e. $\Theta_1 = -\Theta_0)$ will occur at $T_1$ with probability $m(t_1)/\Lambda(t_1; x_0, \theta_0)$.  An algorithm for sampling $Z(t)$ based on this approach is detailed in Algorithm \ref{alg:zigzag2}.

\begin{algorithm}[H]
\renewcommand{\algorithmicrequire}{\textbf{Input:}}
\renewcommand{\algorithmicensure}{\textbf{Output:}}
\caption{Zig-Zag Sampling with thinning\label{alg:zigzag2}}
\begin{algorithmic}[1]
\REQUIRE Initial condition $(x, \theta) \in E$.\\
\ENSURE The event chain $(T_k, X(T_k), \Theta(T_k))_{k=0}^{\infty}$.
\STATE Set $(T_0, X(T_0), \Theta(T_0))=(0, x, \theta)$.
\FOR{$k=0,1,2,\ldots$}
	\STATE{Draw $u \sim U[0,1]$.}
	\STATE{Set $\tau = \widetilde{H}(-\log u)$.}
	\STATE{Set \begin{align*}T_{k+1} &= T_{k} + \tau,\\ 
							 X(T_{k+1}) &= X(T_{k}) + \tau \Theta_{k}.\end{align*}}
	\STATE{With probability $\frac{\lambda(X(T_{k+1}), \Theta_k)}{\Lambda(T_{k+1}; X(T_k), \Theta(T_k))},$   set $\Theta_{k+1} = -\Theta_{k}$ otherwise $\Theta_{k+1} = \Theta_k$.}
\ENDFOR
\end{algorithmic}
\end{algorithm} 

Identifying such a computable upper bound is highly problem specific,  however we can highlight two frequently arising scenarios where upper bounds can be easily constructed.  

\begin{enumerate}
\item Suppose that the log density is globally bounded, i.e. $|U'(x)| \leq K$, for all $x \in \mathbb{R}$.  In this case, we can simply choose $\Lambda(t; x, \theta) = K$.  This case arises in particular for heavy tailed distributions, for example the Cauchy distribution with  $\pi \propto (1 + x^2)^{-1}$.
\item Suppose instead that the second derivative of the log density is absolutely bounded,  i.e. $|U''(x)| \leq L$, for all $x \in \mathbb{R}$.  In this case we have
$$
	\theta U'(x + \theta t) = \theta U'(x) + \int_0^t U''(x + \theta s)\,ds,
$$
so that
$$
	\lambda(x + \theta t, \theta) \leq \max\left(0, \theta U'(x) + L t\right):= \Lambda(t; x, \theta).
$$
For fixed $(x,\theta) \in E$, the integrated intensity function $G(t) = \int_0^t \Lambda(s; x, \theta)\,ds$ has generalised inverse 
$$
\widetilde{H}(z) = \begin{cases}
				\frac{- \theta U'(x) + {\sqrt{2L z}}}{L}\quad &\mbox{ if } \theta U'(x) \leq 0	\\
				\frac{-\theta U'(x) + \sqrt{2L z + U'(x)^2}}{L} &\mbox{ otherwise}.
		   \end{cases}.$$ 
This case arises naturally in various Bayesian inference problems,  in particular logistic regression, see \cite[Section 6.5]{BierkensFearnheadRoberts2016}.
\end{enumerate}

The number of switches that occur in a given time interval will  depend on the intensity function $\Lambda(t; x, \theta)$, and clearly, a poor choice  of this upper bound will cause Algorithm \ref{alg:zigzag2} to undergo many potential switch events which are rejected.   In particular, if the process $Z(t)$ is in stationarity, then the average switching rate will always be higher or equal to that of the direct scheme described in Algorithm \ref{alg:zigzag1}.

\subsection{Computing ergodic averages}
While the event chain $(X(T_k), \Theta(T_k))_{k=0}^\infty$ defines a Markov chain,  it will not be ergodic with respect to the target distribution $\pi$.   To compute an ergodic average for a given observable $f$, the entire continuous time realisation must be used as follows
$$
	\pi_T(f) = \frac{1}{T}\int_0^T f(X_s)\,ds.
$$
Since the Zig-Zag process moves linearly between switches, this can be decomposed into a sum of integrals over straight lines.  Indeed, for $T = T_K$, for some $K$ we have
\begin{equation}
\label{eq:ergodic_average}
	\pi_{T_K}(f) = \frac{1}{T_K}\sum_{k=0}^{K-1} \int_{X({T_k})}^{X({T_{k+1}})}f(x)\,dx = \frac{1}{\sum_{k=0}^{K-1} \tau_k}\sum_{k=0}^{K-1}\int_0^{\tau_k}f(X(T_k) + \Theta(T_k)s)\,ds,
\end{equation}
where $\tau_k = T_{k+1} - T_{k}$.  In many cases,  the integral in~\eqref{eq:ergodic_average} can be computed exactly.  For example, first and $p^{th}$ moment can be computed ergodically via the expressions
$$
	\pi_{T_K}(x) =  \frac{1}{\sum_{k=0}^{K-1} \tau_k}\sum_{k=0}^{K-1} \tau_k X(T_k) + \frac{1}{2}\Theta(T_k)\tau_k^2,
$$
and 
$$
	\pi_{T_K}(x^p) = \frac{1}{\sum_{k=0}^{K-1} \tau_k}\sum_{k=0}^{K-1} \Theta(T_k)\frac{-X(T_k)^{1+p} + (\tau_k \Theta(T_k) + X(T_k))^{1+p}}{(1 + p)}
$$
respectively.  For more complicated observables it will not be possible to evaluate~\eqref{eq:ergodic_average} analytically, and one must resort to some form of quadrature scheme, for example Euler or other higher order methods.


\ack

The authors acknowledge the EPSRC for support under grants EP/D002060/1,  EP/K014463/1 (Joris Bierkens), EP/J009636/1, EP/L020564/1 (Andrew Duncan) as well as EP/K009788/2 (both authors). Furthermore we acknowledge support of the Lloyds Registry Foundation through the Alan Turing Institute. 
We are grateful to the referee and associate editor for useful suggestions with regards to the mathematical exposition, which have certainly helped to improve this paper.

%
%
%
%


\end{document}